\documentclass[twoside]{article} 
\usepackage{fullpage}
\usepackage{amssymb}
\usepackage{amsmath}
\usepackage{amsthm}
\usepackage{hyperref}

\usepackage{nicefrac}
\usepackage{url}
\usepackage{color}
\usepackage[english]{babel}

\usepackage{graphicx}
\graphicspath{{figures/}}

\usepackage{subcaption}
\usepackage{xspace}

\usepackage{todonotes}

\usepackage[inline]{enumitem}
\setlist{nosep}

\usepackage{titlesec}
\titlespacing*{\paragraph}{0pt}{\smallskipamount}{.5em}

\usepackage[ruled,vlined]{algorithm2e}
\SetArgSty{textnormal}
\DontPrintSemicolon
\SetKwFunction{scanside}{scan\_side}
\SetKwFunction{scanface}{scan\_face}
\SetKwFunction{necklacescana}{necklace\_scan\_4a}
\SetKwFunction{necklacescanb}{necklace\_scan\_4b}
\SetKwInOut{Input}{input}
\SetKw{Break}{break}
\SetKw{Continue}{continue}

\newtheorem{theorem}{Theorem}

\newtheorem{claim}{Claim}

\newtheorem{lemma}{Lemma}
\newtheorem{property}{Property}

\newcommand{\starconditionsymbol}{\ensuremath{(\sqsupset)}}
\newcommand{\starcondition}{$\ref{starcond}$\xspace}
\newcommand{\U}{X}
\newcommand{\V}{U}
\newcommand{\W}{W}
\newcommand{\X}{Y}
\newcommand{\rcv}{w}
\newcommand{\tcv}{y}
\newcommand{\lcv}{u}

\newcommand{\calP}{{\ensuremath{\cal P}}\xspace}
\newcommand{\calN}{{\ensuremath{\mathcal N}}\xspace}
\newcommand{\calC}{{\ensuremath{\mathcal C}}\xspace}
\newcommand{\calG}{{\ensuremath{\mathcal G}}\xspace}
\newcommand{\calF}{{\ensuremath{\mathcal F}}\xspace}
\newcommand{\calV}{{\ensuremath{\mathcal V}}\xspace}

\newcommand{\stel}[1]{\ensuremath{#1^\mathrm{s}}\xspace}
\newcommand{\Tint}{{\ensuremath{T_\mathrm{int}}}\xspace}
\newcommand{\Tend}{{\ensuremath{T_\mathrm{end}}}\xspace}
\newcommand{\TSDR}{{\ensuremath{T_\mathrm{SDR}}}\xspace}
\newcommand{\subsubsubsection}[1]{\medskip\noindent{\bf #1}}
\newcommand{\fprev}{{\ensuremath{f_\mathrm{prev}}}\xspace}
\newcommand{\fnext}{{\ensuremath{f_\mathrm{next}}}\xspace}

\renewcommand{\int}{{\ensuremath{\rm int\,}}}

\DeclareMathOperator{\ts}{top}
\DeclareMathOperator{\rs}{right}
\DeclareMathOperator{\ls}{left}
\DeclareMathOperator{\bs}{bottom}
\makeatletter
\newcommand*{\textlabel}[2]{%
  \edef\@currentlabel{#1}% Set target label
  \phantomsection% Correct hyper reference link
  #1\label{#2}% Print and store label
}
\makeatother
\makeatletter
\newcommand*{\caselabel}[1]{%
  \edef\@currentlabel{#1}% Set target label
  \phantomsection% Correct hyper reference link
  \label{case:#1}% store label
}
\makeatother
\newcommand{\reftt}[1]{\texttt{\ref{#1}}}
\newcommand{\refline}[1]{\reftt{line:#1}}

\usepackage{mathtools}
\usepackage{tikz-cd}
\usetikzlibrary{decorations.pathmorphing}
\newcommand\xleadsto[1]{%
    \mathrel{%
        \begin{tikzpicture}[%
            baseline={(current bounding box.south)}
            ]
        \node[%
            ,inner sep=.44ex
            ,align=center
            ] (tmp) {$\scriptstyle #1$};
        \path[%
            ,draw,<-
            ,decorate,decoration={%
                ,zigzag
                ,amplitude=0.7pt
                ,segment length=1.2mm,pre length=3.5pt
                }
            ] 
        (tmp.south east) -- (tmp.south west);
        \end{tikzpicture}
        }
    }

\date{}
    
\newcommand{\repeatcaption}[2]{%
  \renewcommand{\thefigure}{\ref{#1}}%
  \captionsetup{list=no}%
  \caption{#2 (repeated from page \pageref{#1})}%
}

\usepackage{fancyhdr}
\title{Finding Tutte paths in linear time}
\author{
Therese Biedl\thanks{David R.~Cheriton School of Computer
Science, University of Waterloo, Canada.
Work of TB was supported by NSERC. \url{biedl@uwaterloo.ca}}
\and 
Philipp Kindermann\thanks{Lehrstuhl f\"ur Informatik I, Universit\"at W\"urzburg, Germany. \url{philipp.kindermann@uni-wuerzburg.de}}
}

\begin{document}
\addtocounter{page}{-1}
\maketitle
\thispagestyle{empty}
\begin{abstract}
It is well-known that every planar graph has a {\em Tutte path},
i.e., a path $P$ such that any component of $G-P$ has at most three
attachment points on $P$.  However, it was only recently shown that such
Tutte paths can be found in polynomial time.  In this paper, we give a new
proof that 3-connected planar graphs have Tutte paths, which leads to
a linear-time algorithm to find Tutte paths.  Furthermore, 
our Tutte path has special properties: it visits all 
exterior vertices, all components of $G-P$ have exactly three attachment
points, and we can assign distinct representatives to them that are 
interior vertices.  Finally, our running time bound is slightly stronger;
we can bound it in terms of the degrees of the faces that are incident
to $P$.  This allows us to find some applications of
Tutte paths (such as binary spanning trees and 2-walks) in linear time as well.
\end{abstract}

\section{Introduction}

A {\em Tutte path} is a well-known generalization of Hamiltonian paths
that allows to visit only a subset of the vertices of the graph, as long as all
remaining vertices are in components with at most three attachment points.  
(Detailed definitions are below.)  They have been studied
extensively, especially for planar graphs, starting from Tutte's
original result:

\begin{theorem}[\cite{Tutte77}]
\label{thm:2conn}
Let $G$ be a 2-connected planar graph with distinct vertices $\U,\X$ on
the outer face.  Let $\alpha$ be an edge on the outer face. 
Then $G$ has a Tutte path from $\U$ to $\X$ that uses edge $\alpha$. 
\end{theorem}

We refer to the recent work by Schmid and Schmidt
\cite{SS-ICALP18} for a detailed review of the history and applications
of Tutte paths.
It was long not
known how to compute a Tutte path in less than exponential time.  A
breakthrough was achieved by Schmid and Schmidt in 2015 \cite{SS-STACS15,SS18},
when they showed that one can find a Tutte path for 3-connected planar graphs in
polynomial time.  In 2018, the same authors then argued that
Tutte paths can be found in polynomial time even for 2-connected planar
graphs~\cite{SS-ICALP18}.  For both papers, the main insight is to prove the existence of
a Tutte path by splitting the graph into non-overlapping subgraphs to recurse on; 
the split can be found in linear time and therefore the running time becomes
quadratic.

In this paper, we show that Tutte paths can be computed in linear time.
To do so, we give an entirely different proof of the existence of a
Tutte path for 3-connected planar graph.  This proof is very simple if
the graph is triangulated, (we give a quick sketch below), 
but requires more care when faces have larger degrees.
Our path (and also the one by Schmid and Schmidt~\cite{SS-STACS15,SS18}) comes with a 
{\rm system of distinct representatives}, i.e.,
an injective assignment from the components of $G\setminus P$ to vertices
of $P$ that are attachment points.  Such representatives are useful for
various applications of Tutte paths.

Our proof
for 3-connected planar graphs is based on a Hamiltonian-path proof by
Asano, Kikuchi and Saito \cite{AKS84} that was designed to give a
linear-time algorithm; with arguments much as in their paper we can
therefore find the Tutte path and its representatives in linear time.
Since 3-connected planar graphs are (as we argue) the bottleneck in
finding Tutte paths, this shows that the path of Theorem~\ref{thm:2conn}
can be found in linear time.

%%%%%%%%%%%%%%%%%%%%%%%%%%%%%%%%%%%%%%%%%%%%%%%%%%%%%%%%%%%%%%%%%%%%%%%%
\subsection{Preliminaries}

We assume familiarity with graphs, see, e.g., Diestel~\cite{Die12}.  Throughout this
paper, $G=(V,E)$ denotes a graph with $n$ vertices and $m$ edges.  We
assume that $G$ is {\em planar}, i.e., can be drawn in 2D without edge crossings.
A planar drawing of $G$ splits $\mathbb{R}^2$ into connected regions called
{\em faces}; the unbounded region is the {\em outer face} while all others
are called {\em interior faces}.  
A vertex/edge is called {\em exterior} if
it is incident to the outer face and {\em interior} otherwise.
We assume throughout that $G$ is {\em plane}, i.e., one particular abstract 
drawing of $G$ has been fixed
(by giving the clockwise order of edges around each vertex and the edges that
are on the outer face).  Any subgraph of $G$ {\em inherits} this planar
embedding, i.e., uses the induced order of edges and as outer face the face
that contained the outer face of~$G$.  The following notion will be convenient:
Two vertices $v$ and $w$ are {\em interior-face-adjacent} (in a planar graph~$G$)
if there exists an interior face that is incident to both $v$ and $w$.  We will
simply write {\em face-adjacent} since we never consider adjacency via the outer face.

\subsubsubsection{Nooses and connectivity.}
For a fixed planar drawing of $G$, 
let a {\em noose} be a simple closed curve~$\mathcal N$ that goes through vertices and faces and
crosses no edge except at endpoints.  Note that a noose can be described as
a cyclic sequence $\langle x_0,f_1,x_1,\dots,f_s,x_s{=}x_0\rangle$ of vertices and faces 
such that~$f_i$ contains~$x_{i-1}$ and~$x_i$, and hence is independent of
the chosen drawing.  Frequently, the choice of faces will be clear from context
or irrelevant; we then say that $\calN=\langle x_0, \dots,x_s{=}x_0\rangle$ 
{\em goes through} $\{x_1,\dots,x_s\}$.
The subgraph {\em inside/outside} $\mathcal N$ is the graph induced by the vertices
that are on or inside/outside $\mathcal N$.  The subgraph {\em strictly inside/outside} 
is obtained from this by deleting the vertices on \calN.

A graph $G$ is {\em connected} if for any two vertices $v,w$ there is a path
from $v$ to $w$ in $G$.
A {\em cutting $k$-set} in~$G$ is a set $S=\{x_1,\dots,x_k\}$ of
vertices such that $G\setminus S$ has more connected components than $G$.
We call it a {\em cutting pair} for $k=2$ and
a {\em cutting triplet} for $k=3$.
A graph $G$ is called {\em $k$-connected} if it has no cutting $(k-1)$-set.
Since we are only studying planar graphs, it will be convenient to use a
characterization of connectivity via nooses. 
Consider a noose $\mathcal N$ that goes through $\{x_1,\dots,x_k\}$
(and no other vertices), and there are vertices both strictly inside and 
strictly outside \calN.    Then clearly $S=\{x_1,\dots,x_k\}$ is a cutting 
$k$-set.  Vice versa, it is not hard to see that 
in a planar graph,
any cutting $k$-set~$S$ for $k=1,2,3$ gives rise to a noose \calN through $S$
that has vertices both strictly inside and strictly outside.
A {\em cut component}~$C$ of~$S$ is a subgraph strictly inside a noose \calN through
some of the vertices of $S$
such that $C$ contains at least one vertex not in $S$ and is inclusion-minimal
among all such nooses. In particular, a cut component $C$ contains
no vertices or edges of~$S$, but in our algorithm
we will frequently add the vertices and edges of~$S$ to $C$ and call the result~$C^+$.

\subsubsubsection{Hamiltonian paths and Tutte paths.}
A {\em Hamiltonian path} is a path that visits every vertex exactly once.
To generalize it to Tutte paths, we need more definitions.
Fix a path~$P$ in the graph.
A {\em $P$-bridge}~$C$ is a cut component of~$P$; its {\em attachment
points} its vertices on $P$.%
\footnote{Our definition of $P$-bridge considers only the {\em proper}
$P$-bridges \cite{Tutte77} that contain at least one vertex.}
A {\em Tutte path} is a path $P$ such that any $P$-bridge~$C$ has at most
three attachment points,
and if $C$ contains exterior edges, then it has at most two attachment
points.  Our Tutte paths for 3-connected graphs will be such that
no $P$-bridges contain exterior edges, so the
second restriction holds automatically.

A Tutte path with a {\em system of
distinct representatives} (SDR), also called a {\em $\TSDR$-path} for short,
is a Tutte path~$P$ together with an injective
assignment $\sigma$ from the
$P$-bridges to vertices in~$P$ such that for every $P$-bridge $C$
vertex $\sigma(C)$ is an attachment point of $C$.

Given a path $P$ in a plane graph, we denote by $F(P)$ the set of all interior 
faces that contain at least one vertex of $P$.

\subsection{From 3-connected to 2-connected}

In this section, we show that, to find the path of Theorem~\ref{thm:2conn} 
efficiently, it suffices to consider
3-connected planar graphs.  

We re-prove Theorem~\ref{thm:2conn}, presuming it holds for 3-connected
planar graphs, by induction on the number of vertices with an inner
induction on the number of exterior vertices.  Say we want to find a Tutte path
from $\U$ to $\X$ that uses exterior edge $\alpha=(\V,\W)$, where $\U,\X$
are exterior vertices.  In the base case, $G$ is 3-connected
and we are done.  So assume that $G$ has cutting pairs. 
If edge $(\U,\X)$ does not exist, then add it in such a way that $\alpha$
stays exterior, and find a Tutte path $P$ in
the resulting graph recursively (it has fewer exterior vertices).
Since $\{\U,\X\}\neq \{\V,\W\}$ (because $(\V,\W)\in G$ while 
$(\U,\X)\not\in G$), path $P$ visits at least one vertex other than $\U,\X$,
and so cannot use edge $(\U,\X)$.  So it is also a Tutte path of $G$.

Now, assume that $(\U,\X)$ exists.  
Repeatedly split the graph at any cutting pair $\{u,v\}$
into cut components $C_1,\ldots,C_k$, and store the \emph{3-connected
components} $C_1^+,\ldots,C_k^+$---induced by the cut components and~$u,v$
and inserting a virtual edge $(u,v)$---in a so-called \emph{SPQR-tree}~\cite{DBT89,GM00},
which additionally creates one leaf node for every edge of $G$.
This can be done in linear time~\cite{HT73}.

Root the SPQR-tree at the node of edge $(\U,\X)$.
For each 3-connected component $C^+$ other than the root, 
set $\{\U_C,\X_C\}$ to be the
cutting pair that $C^+$ has in common with its parent component, and
observe that these two vertices are necessarily exterior in~$C$
since $\U,\X$ are exterior in $G$; see Fig.~\ref{fig:spqr}.

If $C^+$ has only these two vertices, then let $P_C$ be the path $(\U_C,\X_C)$.
Otherwise, define an edge $\alpha_C\neq (\U_C,\X_C)$ of $C^+$ as follows:  If
the node of~$\alpha$ is a descendant of $C^+$, then let $\alpha_C$ be the
virtual edge of $C^+$ that it shares with the child that leads to this 
descendant.   Note that~$\alpha_C$ is a virtual edge, and it is necessarily 
on the outer face of $C$
since $\alpha$ is on the outer face of~$G$.  Otherwise ($\alpha$ is not in
a descendant of $C^+$) choose $\alpha_C$ to be an arbitrary exterior edge of $C$
other than $(\U_C,\X_C)$.  Let $P_C$ be a Tutte path that
begins at~$\U_C$, ends at~$\X_C$ and uses edge $\alpha_C$; we know that this
exists since $C^+$ is either a triangle or a 3-connected graph.

\begin{figure}[t]
  \centering
  \includegraphics[page=5]{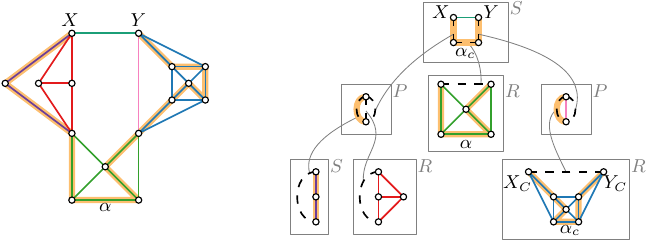}
  \caption{A 2-connected graph, its SPQR-tree (leaf nodes are omitted), and
    its Tutte path.}
  \label{fig:spqr}
\end{figure}

Now, obtain the Tutte path $P$ of $G$ by repeatedly substituting paths of
3-connected components.  Specifically, initiate $P$ as the virtual
copy of edge $(\U,\X)$ that was added when we created the node for $(\U,\X)$.  
For as long as $P$ contains a virtual edge $(u,v)$,
let $C^+$ be the child component at this virtual edge and observe that
$\{\U_C,\X_C\}=\{u,v\}$.
{\em Substitute $P_{C}$ in place of edge $(u,v)$ of $P$},
i.e., set $P$ to be $\U \xleadsto{P} u/v \xleadsto{P_{C}}
v/u \xleadsto{P} \X$.    Note that, if $C^+$ is not the singleton-edge $(u,v)$,
then $P_C$ contains $\alpha_C$, which is a virtual edge.  This means that the 
process repeats until we have substituted the real edges from the leaves of
the SPQR-trees. In particular (due to our choice of $\alpha_C$),
we will substitute the paths from all components between $(\U,\X)$ and 
$\alpha$, which means that $\alpha$ is an edge of the final path $P$ as required.

Observe that for some 3-connected components we do not substitute
their paths; these become $P$-bridges with two attachment points.
There may also be some $P$-bridges within each 3-connected components,
but these have at most three attachment points since we used a Tutte path
for each component.  So the result is the desired Tutte path.
Since we compute one Tutte path per
3-connected component, and this can be done in time proportional
to the size of the component, the overall running time is linear.

\subsection{Simple proof for triangulated planar graphs}
\label{sec:triangulated}

So it suffices to find Tutte paths in linear time for 3-connected
planar graphs.
As a convenient warm-up, we sketch here first the (much simpler) case
of a {\em triangulated} planar graph $G$, i.e., every face is a triangle.
Remove, for every non-facial
triangle $T$, the graph in its interior.
What remains is a 4-connected
triangulated planar graph $H$, say it has $k$ vertices. The graph~$H$ has a Hamiltonian 
path $P$, and  Asano, Kikuchi, and Saito \cite{AKS84} showed how to find it in 
linear time.  Studying their proof, one can easily verify that we can force that~$P$
begins at $\U$, ends at $\X$, and contains edge $\alpha$, for given
$\U,\X,\alpha$ on the outer face.%
\footnote{For readers familiar with \cite{AKS84}: It suffices to choose
$A,B,C$ in the outermost recursion such that $A=\U$, $B=\X$ and $C$ is the 
end of $\alpha$ that is closer to $B$.}
Note that $P$
automatically is a Tutte path, because every $P$-bridge resides inside an
interior face of $H$, and hence has three attachment points and no
edge on the outer face.

To find a system of representatives, we initially allow edges of $P$ to be 
representatives,
but forbid using exterior vertices or the edge $\alpha$.
Thus, assign to every interior face $T$ of $H$ a representative
$\sigma(T)\in V(P)\cup E(P)$ such that $\sigma(T)$ is
incident to $T$, not an exterior vertex or $\alpha$, and no two triangles obtain the 
same representative.
There are $2k-5$ faces of $H$,
and path $P$ has $k-2$ edges $\neq \alpha$ and $k-3$ interior vertices, 
so there are sufficiently many possible representatives.  One can argue
that for the Hamiltonian path from \cite{AKS84}, it is possible to assign
the faces to these representatives in an injective manner
(we omit the details; they are implicit in our proof below.) 

To obtain a $\TSDR$-path,
we apply a {\em substitution trick} (explained in detail below) for every
edge~$e\in P$ that is used as representative $\sigma(T)$ for some interior face 
$T$.   Namely, if $T$ has no
$P$-bridge inside, then simply remove $e$ as representative.  Otherwise,
replace $e$ by a (recursively obtained) $\TSDR$-path $P_T$ of 
the $P$-bridge inside $T$.  Since we can specify the ends of $P_T$,
and it uses no vertices of its outer face $T$ as representatives, this gives
a $\TSDR$-path of $G$ after repeating at all edge-representatives.
The overall running time is linear, because we can find $H$ by computing
the tree of 4-connected components in linear time~\cite{Kant97},
and then compute a Hamiltonian path in 
each 4-connected component in linear time \cite{AKS84}.  
The $\TSDR$-path is obtained by
substituting the Hamiltonian paths of child-components into the
one of the root as needed to replace edge-representatives.

%%%%%%%%%%%%%%%%%%%%%%%%%%%%%%%%%%%%%%%%%%%%%%%%%%%%%%%%%%%%%%%%%%%%%%%%
\section{Tutte paths in 3-connected planar graphs}
\label{sec:TuttePath}

For triangulated planar graphs, one can quite easily find a $\TSDR$-path
by removing the interiors of all separating triangles, and finding for
the resulting 4-connected planar graph a Hamiltonian path using the
approach of Asano, Kikuchi, and Saito~\cite{AKS84}.  It is not hard to
see that we can assign representatives to all separating triangles,
possibly after expanding the path using the substitution trick described
below.  (We omit the details for space reasons.)

For 3-connected planar graphs that are not triangulated, we use the same
approach, but must generalize many definitions from~\cite{AKS84}
and add quite a few cases  because now face-adjacent vertices
are not necessarily adjacent.  To keep the proof self-contained, we re-phrase
everything from scratch.%
\footnote{Indeed, due to attempts to simplify the notations similar as
done in \cite{BD-JoCG16}, the
reader familiar with \cite{AKS84} may barely see the correspondence between the
proof and~\cite{AKS84}.  Roughly, their Condition (W) corresponds to $c3c(\U,\W,\X)$,
their Case 1 is our Case~\ref{case:3a}, and their Case 3 combines our Case~\ref{case:2} with
our Case~\ref{case:4a} (but resolves it in a symmetric fashion).}
 
We need a few definitions. 
The {\em outer stellation} of a planar graph $G$ is the graph obtained
by adding a vertex in the outer face and connecting it to all exterior vertices. 
A planar graph~$G$ is called {\em internally 3-connected} if its outer stellation
is 3-connected. Note that this implies that $G$ is 2-connected, 
any cutting pair is {\em exterior} (i.e., has both vertices on the outer face)
and has only two cut components that contain other vertices.
In the following, we endow~$G$ with~$k$ {\em corners}, which
are $k$ vertices $X_1,\dots,X_k$ that appear in this order on the 
outer face.    
Usually, $k=3$ or $4$, but occasionally we allow larger $k$.  
A {\em side} of such a graph is the outer face
path between two consecutive corners that does not contain any other corners.
The {\em corner stellation} $\stel{G}$ is obtained by
adding a vertex in the outer face and connecting it to the corners.
We say that $G$ is {\em corner-3-connected with respect to corners
$X_1,\dots,X_k$}
(abbreviated to ``{\em $G$ satisfies $c3c(X_1,\dots,X_k)$}'')
if~$\stel{G}$ is 3-connected.
Figure~\ref{fig:corner3con} illustrates this condition.  
It is easy to show that $G$ satisfies $c3c(X_1,\dots,X_k)$
if and only if $k\geq 3$, $G$ is internally 3-connected,
and no cutting pair $\{v,w\}$ of $G$ has both~$v$ 
and~$w$ on one side of $G$.

For ease of proof we make the induction hypothesis stronger than just
having a $\TSDR$-path, by restricting which vertices {\em must} be
visited and which vertices {\em must not} be representatives.  
A \emph{$\Tint$-path} is a $\TSDR$-path $P$ that visits all exterior vertices, and
where representative $\sigma(C)$ is interior, for all $P$-bridges $C$.
The goal of the remainder of this section is to prove the following result
(which immediately implies Theorem~\ref{thm:2conn} for 3-connected graphs%
\footnote{Theorem~\ref{thm:2conn} allows $(\V,\W)=(\U,\X)$, but then holds
trivially since using edge $(\U,\X)$ as path satisfies all conditions.  We
require $(\V,\W)\neq (\U,\X)$ since we want not just a Tutte path but a
$\Tint$-path, and the single-edge path $(\U,\X)$ would 
allow only exterior vertices as representatives.}%
):

\begin{lemma}
\label{lem:T_SDR}
\label{lem:3conn}
\label{lem:T_int}
Let $G$ be a plane graph with distinct vertices $\U,\X$ on
the outer face.  Let $(\V,\W)\neq(\U,\X)$ be an edge on the outer face.
If $G$ satisfies $c3c(\U,\V,\W,\X)$,
then it has a $\Tint$-path that begins at $\U$, ends at $\X$, and contains $(\V,W)$.
\end{lemma}

We need a second result for the induction.  Let a
\emph{$\Tend$-path} be a $\TSDR$-path $P$ that visits all exterior vertices, and
where representative $\sigma(C)$ is interior or the last vertex of $P$, for all $P$-bridges $C$.

\begin{lemma}
\label{lem:T_end}
Let $G$ be a plane graph with distinct vertices $\U,\X$ on
the outer face.  Let $(\V,\W)\neq(\U,\X)$ be an edge on the outer face.
If $G$ satisfies $c3c(\U,\V,\W,\X)$
and 
\begin{eqnarray*}
\textlabel{\starconditionsymbol}{starcond}\quad & (\W,\X) \textrm{ and } (\X,\U) \textrm{ are edges,}
\end{eqnarray*}
then $G$ has a \Tend-path~$P$ that begins at $\U$, ends at $\X$, and uses $(\V,\W)$
and $(\W,\X)$.%
\footnote{This lemma is a special case of the ``Three Edge Lemma'' 
\cite{TY94}, which states that for any three edges on the outer face
there exists a Tutte cycle containing them all.  However, it cannot
simply be obtained from it since we require restrictions on the
location of representatives.}
Further, if $\X$ is the representative of a $P$-bridge $C$,
then~$C$ has~$\W$ and~$\X$ as attachment points.
\end{lemma}

See Figure~\ref{fig:case4a-pplus} for a graph that satisfies~\starcondition.

We assume throughout that
$\U,\V,\W,\X$ are enumerated in ccw order along the outer face, 
the other case can be resolved by reversing the planar embedding.  

The following trick will help shorten the proof:  If
graph $G$ satisfies~\starcondition, then
Lemma~\ref{lem:T_end} implies Lemma~\ref{lem:T_int}. 
Namely, if Lemma~\ref{lem:T_end} holds 
then we have a \Tend-path $P$ from~$\U$ to~$\X$ through $(\V,\W)$ and $(\W,\X)$.  
If this is not a \Tint-path, then
some $P$-bridge~$C$ has $\X$ as representative, and by assumption
also has $\W$ as attachment point.  It must have a third attachment point $u$,
otherwise $\{\W,\X\}$ would be a cutting pair within one side of $G$, contradicting
corner-3-connectivity.  It has no more attachment points since~$P$ is a Tutte path,
so $\{\W,\X,u\}$ is a cutting triplet.  We apply the {\em substitution trick}
described below (and useful in other situations as well),  which 
replaces $(\W,\X)$ with a path through~$C$ that does not
use $u$.  Thus, $C$ no longer needs a representative and we obtain a
\Tint-path.

\subsubsubsection{The substitution trick.}
This trick can be applied whenever we 
have an edge $e=(w,y)$ used by some $\TSDR$-path $P$, 
and a $P$-bridge $C$ that resides inside a noose through
some cutting triplet $\{u,w,y\}$ for some vertex $u$.
Define
$C^+=G[C] \cup  \{(u,w),(w,y)\}\setminus \{(u,y)\}$,
where edges are added only if they did not exist in $G[C]$.%
\footnote{We apply the substitution trick even
when $V(C^+)=V(G)$ and~$G$ has a triangular outer face; not adding edge $(w,y)$ will 
ensure that $C^+$ has fewer interior vertices and induction can
be applied.}

\newcommand{\ctcsubsCaption}{(a)~Corner-3-connectivity $c3c(\U,\V,\W,\X)$, (b)~the substitution trick, and (c) Case~\ref{case:1}.}
\begin{figure}[t]
\centering
\subcaptionbox{\label{fig:corner3con}}{\includegraphics[scale=.75]{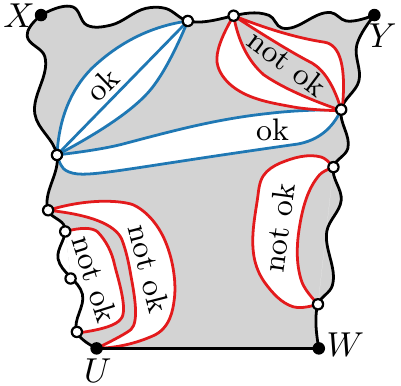}}
\hfil
\subcaptionbox{\label{fig:substitution}}{\includegraphics[scale=.9]{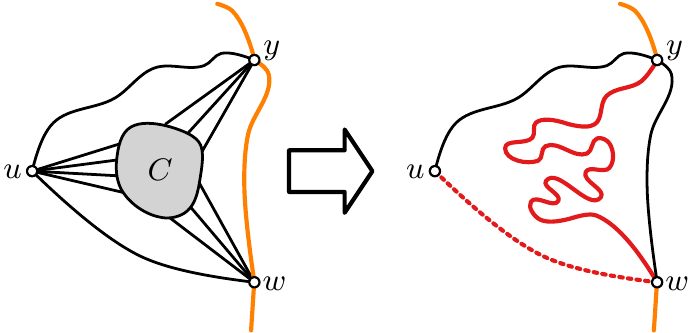}}
\hfil
\subcaptionbox{\label{fig:case1}}{\includegraphics{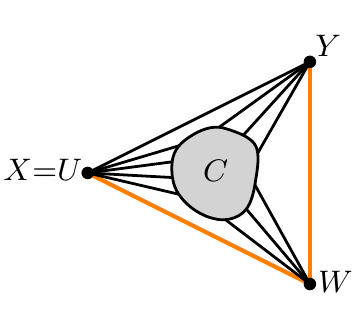}}
\caption{\ctcsubsCaption}
\label{fig:c3csubs}
\end{figure}

One easily verifies that $C^+$ satisfies $c3c(u,v,w)$, else there would have 
been a cutting pair in~$G$.
Hence, by induction, $C^+$ has a \Tint-path $P_{C^+}$ from $u$ to
$y$ that uses edge $(u,w)$.  It does not use the edge $(u,y)$
since $P_{C^+}$ begins at $u$ with edge $(u,w)$.  So $P_{C^+}\setminus (u,w)$ 
is a path in $C$ from $w$ to $y$ that does not visit $u$.   
Substitute this in place of edge $(w,y)$ of $P$; see Fig.~\ref{fig:substitution}.  
One easily verifies that the resulting path $P'$ is a \Tint-path.  We will
prove this in full detail in Section~\ref{sec:additional}, but roughly
speaking, combining paths preserves \Tint-paths because every 
$P'$-bridge can inherit its representative from~$P$ or~$P_{C^+}$,
and no vertex is used twice as representative since $P_{C^+}$ does
not use $\{u,v,w\}$ as representatives.

\subsection{Proof of Lemma~\ref{lem:T_int} and Lemma~\ref{lem:T_end}}
\label{subsec:proof}

We prove the two lemmas simultaneous 
by induction on the number of vertices of $G$, with an
inner induction on the number of interior vertices.  The base case is $n=3$
where $G$ is a triangle, but 
the same construction works whenever the outer face is a triangle (see below).
For the induction step, we need the notation
$S_{xy}$, which is the outer face path from $x$ to $y$ in ccw 
direction.  In particular, the four sides are $S_{\U\V}$, $S_{\V\W}$, $S_{\W\X}$ and
$S_{\X\U}$.  We sometimes name sides as
suggested by Fig.~\ref{fig:corner3con}, so $S_{\U\V}, S_{\V\W}, S_{\W\X}$
and $S_{\X\U}$ are the {\em left/bottom/right/top side}, respectively.

\subsubsection{\texorpdfstring{Case 1: The outer face is a triangle}{Case 1}}\caselabel{1}

Figure~\ref{fig:case1} illustrates this case.
We know that $\U\neq \X$ and $\V\neq \W$, so we must have $\U=\V$ or $\W=\X$.
For Lemma~\ref{lem:T_end}, we know that
\starcondition holds, which forces $\W\neq \X$, hence $\U=\V$.
For Lemma~\ref{lem:T_int}, we may assume $\U=\V$ by symmetry, for otherwise
we reverse the planar embedding, find a path from $Y$ to $X$ that uses
$(W,U)$ (with this, we have $\U'=\V'$) and then reverse the result.

So $\U=\V$.  Define $P$ to be $\langle \U{=}\V,\W,\X\rangle$ 
and observe that this is a 
\Tend-path, because
the unique $P$-bridge $C$ (if any) has attachment points
$\{\V,\W,\X\}$, and we can assign $\X$ to be its representative. 
So Lemma~\ref{lem:T_end} holds.
Since condition \starcondition is satisfied, this implies Lemma~\ref{lem:T_int}.

\subsubsection{\texorpdfstring{Case 2: $G$ has a cutting pair $\{\lcv,\rcv\}$
with $\lcv$ and $\rcv$ on the left and right side}{Case 2}}\caselabel{2}

Figure~\ref{fig:case2} illustrates this case.
Let $\calN$ be a noose through $\lcv$ and $\rcv$ along a common interior
face~$f^*$ and then going through the outer face.
Let $G_t$ and $G_b$
be the subgraphs inside and outside $\calN$, named such that $G_b$
contains the bottom side. Let $G_t^+/G_b^+$ be the graphs obtained from
$G_t/G_b$ by
adding $(\lcv,\rcv)$ to each, 
even if it did not exist in $G$ (we will ensure that the final path
does not use it).  

We first show Lemma~\ref{lem:T_int}.
One can easily verify that $G_t$ satisfies $c3c(\U,u,w,\X)$ since its
outer face is a simple cycle; see Appendix~\ref{sec:additional}.
Apply induction and find a \Tint-path
$P_t$ of $G^+_t$ from $\U$ to $\X$ that uses edge $(\lcv,\rcv)$.
Now apply a modified substitution trick to $(\lcv,\rcv)$.  Namely,
by induction, there is a \Tint-path~$P_b$ of $G^+_b$ from $\lcv$ to $\rcv$ that 
uses edge $(\V,\W)$.
Substitute $P_b$ into $P_t$ in place of $(\lcv,\rcv)$ to get $P$.
Path $P$ uses $(\V,\W)$ since $P_b$ does.   It does not use $(\lcv,\rcv)$
since we removed this from $P_t$, and 
$P_b$ cannot use it since 
$P_b$ starts at $\lcv$,
ends at $\rcv$, and visits $(\V,\W)$ in between.  
So after inheriting representatives from $P_b$ and $P_t$ we obtain a \Tint-path~$P$ in~$G$.

To prove Lemma~\ref{lem:T_end}, note that exactly one of $G^+_t$ and $G^+_b$
contains $(\W,\X)$; use a \Tend-path~for this subgraph and create $P$ as above.
Only one graph uses $\X$ as representative, and 
one easily
shows that~$P$ is a \Tend-path.

\newcommand{\casetwoCaption}{(a)~Case~\ref{case:2}, (b)--(d) proof of Lemma~\ref{lem:3conn} for Case~\ref{case:2}}
\begin{figure}[t]
\centering
\subcaptionbox{Case~\ref{case:2}}{\includegraphics[page=1,scale=.72]{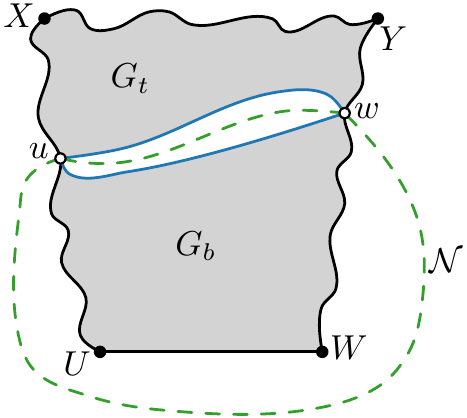}}
\hfil
\subcaptionbox{The path~$P_t$}{\includegraphics[page=2,scale=.72]{case2}}
\hfil
\subcaptionbox{The path~$P_b$}{\includegraphics[page=3,scale=.72]{case2}}
\hfil
\subcaptionbox{The path~$P$}{\includegraphics[page=4,scale=.72]{case2}}
\caption{\casetwoCaption}
\label{fig:case2}
\end{figure}

\subsubsection{\texorpdfstring{Case 3: $G$ has a cutting pair $\{\tcv,\rcv\}$ with
$\tcv$ and $\rcv$ on the top and right side, respectively.  Furthermore, 
there is an interior face $f^*$ containing $\tcv$ and $\rcv$ that does not contain $\X$.}{Case 3}}\caselabel{3}

For later applications, we first want to point out that if $G$ has a
cutting pair $\{\tcv,\rcv\}$ on the top and right side for which $(\tcv,\rcv)$
is an edge, then such a face $f^*$ always exists, because
there are two interior faces containing $\tcv$ and $\rcv$, 
and not both can contain $\X$.

\smallskip

Figure~\ref{fig:case3} illustrates this case.  
We know $\rcv\neq \X\neq \tcv$, else $\{\tcv,\rcv\}$ would be a cutting pair within one side.
We may assume $\tcv\neq \U$; else we can use Case~\ref{case:2}.
Hence the top side contains at least three vertices~$\U,\tcv,\X$,
so \starcondition does not hold and we have to prove only Lemma~\ref{lem:T_int}.

We choose $\{\tcv,\rcv\}$ such that 
$\rcv$ is as close to $\W$ as possible (along the right side).
The face $f^*$ containing $\tcv,\rcv$ may have multiple edges on 
the top side; let $\tcv$ be the one that is as close to $\X$ as possible.
Define $G_b$, $G_b^+$, $G_t$, $G_t^+$ to be as in Case~\ref{case:2}.
Since the outer face of $G^+_b$ is a simple cycle, it
satisfies $c3c(\U,\V,\W,\rcv,\tcv)$.
But since we chose~$\rcv$ to be as close to~$\W$ as possible, it also satisfies $c3c(\U,\V,\W,\tcv)$.
Namely, assume for contradiction that some cutting pair $\{\tcv',\rcv'\}$ 
exists along the side $S_{\W \rcv}\cup (\rcv,\tcv)$ of $G_b^+$;  see Fig.~\ref{fig:case3-c3c}.
Since there is no cutting pair 
within~$S_{\W \rcv}$, 
it must have the form $\{\tcv,\rcv'\}$ for some $\rcv'\neq w$ on $S_{Ww}$.
As $f^*$ does not contain $\X$, neither
can any face containing $\{\tcv,\rcv'\}$, so $\{\tcv,\rcv'\}$ could have been used
for Case~\ref{case:3}, contradicting our choice of~$\rcv$. 

By induction, we can find a \Tint-path $P_b$ of $G^+_b$ from $\U$ to $\tcv$
that includes the edge $(\V,\W)$. 
The plan is to combine~$P_b$ 
with a path through $G_t$, but we must distinguish cases.

\newcommand{\casethreeCaption}{Case~\ref{case:3}: (a)~$G_b^+$ satisfies $c3c(\U,\V,\W,\tcv)$, (b) Case~\ref{case:3a},
(c) Case~\ref{case:3b-1}, (d) Case~\ref{case:3b-2}.}
\begin{figure}[t]
\centering
\subcaptionbox{\label{fig:case3-c3c}}{\includegraphics[scale=.75,page=2]{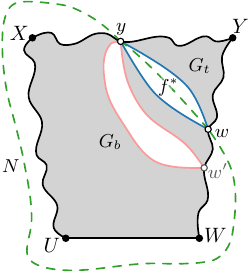}}
\hfil
\subcaptionbox{\label{fig:case3-a}}{\includegraphics[scale=.75,page=8]{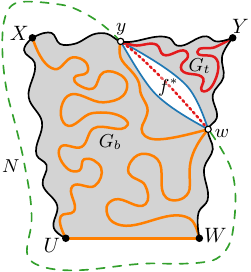}}
\hfil
\subcaptionbox{\label{fig:case3b-gt}}{\includegraphics[scale=.75,page=5]{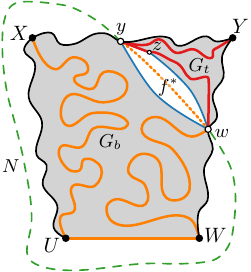}}
\hfil
\subcaptionbox{\label{fig:case3b-2}}{\includegraphics[scale=.75,page=6]{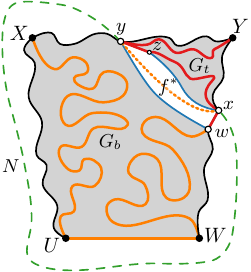}}
\caption{\casethreeCaption}
\label{fig:case3}
\end{figure}

\subsubsubsection{Case 3a: $P_b$ does not contain $(\tcv,\rcv)$ 
or $(\tcv,\rcv)\in G$.}\caselabel{3a}
Observe that $G_t^+$ satisfies $c3c(\tcv,\rcv,\X)$.  By induction, find a \Tint-path $P_t$ in $G_t^+$
from $\X$ to $\rcv$ that uses edge $(\tcv,\rcv)$. Append
the reverse of $P_t\setminus (\tcv,\rcv)$ to $P_b$ to obtain a 
\Tint-path; see Fig.~\ref{fig:case3-a}.

\subsubsubsection{Case 3b: $P_b$ contains $(\tcv,\rcv)$ and $(\tcv,\rcv)\not\in G$.}\caselabel{3b}
In this case, we must remove $(\tcv,\rcv)$ from the path and hence
use a subpath in $G_t$ to reach vertex $\tcv$.
This requires further subcases.  Let $\pi_f$ be the path along $f^*$ from $\tcv$ to $\rcv$ that
becomes part of the the outer face of $G_t$.
Let $(y,z)$ be the edge incident to~$y$ on~$\pi_f$.

\smallskip
\noindent\textbf{Case 3b-1: $\pi_f$ contains no vertex on the outer face of
$G$ other than $\tcv$ and $\rcv$.}\caselabel{3b-1} See Fig.~\ref{fig:case3b-gt}.
The outer face of $G_t$ is then a simple cycle and $G_t$ satisfies $c3c(\rcv,\tcv,\X)$. 
By induction, we can find a $\Tint$-path $P_t$ in $G_t$ that begins at
$\X$, ends at $\rcv$, and uses $(\tcv,z)$.

\smallskip
\noindent\textbf{Case 3b-2: $\pi_f$ contains a vertex $x\neq \tcv,\rcv$ on the outer 
face of $G$.}\caselabel{3b-2} See Fig.~\ref{fig:case3b-2}.
Since~$x$ is on $f^*$, it cannot be on the top side by choice of $\tcv$. 
So $x\in S_{\rcv\X}\setminus {\X}$.  In fact, $x$ must be the neighbor of $\rcv$ on both
$S_{\rcv\X}$ and $\pi_f$, else there would be a cutting pair within the
right side.
Set $G_t'$ to be the graph inside a noose through $\tcv$ and $x$ that has $\X$ inside.
Since $\pi_f$ has no vertices other than $\tcv,x,\rcv$ on the outer face of $G$, graph
$G_t'$ has a simple cycle as outer face,
and therefore satisfies $c3c(\X,\tcv,x)$, 
so it satisfies $c3c(\X,\tcv,z,x)$. By induction,
we can find a $\Tint$-path~$P_t'$ of $G_t$ that begins at $\X$, ends at~$x$,
and uses $(\tcv,z)$.
We append $(\rcv,x)$ to obtain~$P_t$.

\smallskip
In both cases, we obtain a path $P_t$ that begins at $\X$, ends at $\rcv$, and visits all
of $G_t$.  Appending the reverse of this to $P_b\setminus (\tcv,\rcv)$ gives the \Tint-path.

\subsubsection{\texorpdfstring{Case 3$'$: $G$ has a cutting pair $\{\tcv,\rcv\}$ with
$\tcv$ and $\rcv$ on the top and left side, respectively.  Furthermore, 
there is an interior face $f^*$ containing $\tcv$ and $\rcv$ that does not contain $\U$.}{Case 3'}}\caselabel{3'}

This is handled symmetrically to Case 3.

\subsubsection{\texorpdfstring{Case 4:  None of the above}{Case 4}}\caselabel{4}

In this case, we split $G$ into one big graph $G_0$ and (possibly many) smaller
graphs $G_1,\dots,G_s$, recurse in $G_0$ and then substitute \Tint-paths of
$G_1,\dots,G_s$ or use them as $P$-bridges.

We need two subcases, but first give some steps that
are common to both.
Let $\X_{\U}$ be the neighbor of~$\X$ on the top side.
Define a {\em $B$-necklace} (for $B\in \{\V,\W\}$)
to be a noose  $\calN_0: \langle \X_{\U}{=}x_0,f_1,x_1,\dots,x_{s-1},\allowbreak f_s,x_s{=}B,f_o\rangle$,
(where $f_{o}$ is the outer face) for which
$x_i$ is face-adjacent to at least one vertex 
on $S_{\W\X}\setminus \{B\}$ for $1\leq i\leq s-1$. 
See also Fig.~\ref{fig:case4a}.
We say that the necklace
is {\em simple} if it contains no vertex twice,
and {\em interior} if every $x_i$ (for $0 < i < s$) is an interior vertex.
One can argue that if none of the previous cases applies, then
there always exists a simple interior $B$-necklace
(see Section~\ref{sec:additional}).

Route $\calN_0$ through the outer face such that 
the left side is in its interior, and let $G_0$ (the ``left graph'')
be the graph inside $\calN_0$.  We say that $\calN_0$ is {\em leftmost} if
(among all simple interior $B$-necklaces) its left graph~$G_0$ is smallest,
and (among all simple interior $B$-necklaces whose left graph 
is $G_0$) it contains the most vertices of $G_0$.
Fix a leftmost $B$-necklace $\langle x_0,\dots,x_s\rangle$.

\newcommand{\casefouraCaption}{Case~\ref{case:4}. (a)~A simple interior $\V$-necklace that is not leftmost due to face~$f$
(which yields a cutting pair $\{x_j,x_i\}$), and since it
could include vertex $z$.  (b)~The graphs $G_1,\ldots,G_s$.
(c--d)~Case~\ref{case:4a}. The path~$P^+$ after using the substitution trick and
(d)~assignment of the representatives.}
\begin{figure}[t]
\centering
\subcaptionbox{\label{fig:necklace}}{\includegraphics[page=1,scale=.75]{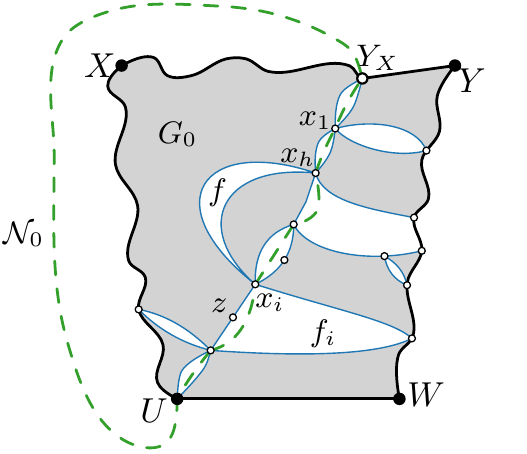}}
\hfil
\subcaptionbox{\label{fig:construction}}{\includegraphics[page=2,scale=.75,trim={8mm 0 0 0},clip]{necklace}}
\hfil
\subcaptionbox{\label{fig:case4a-pplus}}{\includegraphics[page=3,scale=.75]{case4a}}
\hfil
\subcaptionbox{\label{fig:case4a-rep}}{\includegraphics[page=4,scale=.75]{case4a}}
\caption{\casefouraCaption}
\label{fig:case4a}
\end{figure}

\begin{claim}
\label{cl:interior}
\label{cl:inner}
\label{cl:virtual}
If $(x_i,x_{i+1})$ is not an edge for some $0\leq i < s$, 
then the face $f_i$ of $\calN_0$ contains no vertex of $S_{\W\X}\setminus\{B\}$.
\end{claim}
\begin{proof} 
If $(x_i,x_{i+1})$ is not an edge of $G$, then both paths
from $x_i$ to $x_{i+1}$ on $f_i$ contain at least one other vertex.
One of them, say~$z$, is inside $\calN_0$. If~$f_i$ contains a vertex of 
$S_{\W\X}\setminus \{B\}$,
then~$x_i$ and~$z$ are face-adjacent,~$z$ and $x_{i+1}$ are face-adjacent, and
$z$ has a neighbor on $S_{\W\X}\setminus \{B\}$, so $x_0,\ldots,x_i,z,x_{i+1},\ldots,x_s$
is a simple interior $B$-necklace with the same left graph but containing more vertices of $G_0$.  Hence $\calN_0$ is not leftmost, a contradiction.
\end{proof}

For $i=0,\dots,{s-1}$, let $t_i$ be the vertex on $S_{\W\X}\setminus \{B\}$ that 
is face-adjacent to $x_i$ and closest to $\X$ (along the right side) among all such vertices.  
Set $t_s=\W$ if $x_s=\V$, and $t_s=\X$ otherwise.
For $0< i\leq s$, define $\mathcal{N}_i$ to be the noose through
$\langle x_{i-1},x_i,t_i,t_{i-1}\rangle$ such that the left side is outside $\calN_i$.
For $0< i\leq s$ let $G_i$ be the graph inside $\calN_i$
(i.e., a cut component of $\{x_{i-1},x_i,t_{i-1},t_i\}$);
see Fig.~\ref{fig:construction}.

Let $G^+$ be the graph obtained from $G$ by adding
{\em virtual edges} $(x_i,x_{i-1})$ and $(t_i,t_{i-1})$ 
(for $i=1,\dots,s$) whenever these two vertices
are distinct and the edge did not exist in $G$.  
Let~$G_0^+$ be the graph obtained from $G_0$ by likewise adding virtual edges 
$(x_0,x_1),\dots,(x_{s-1},x_s)$.  This makes the outer face of $G_0$ a simple
cycle, so~$G_0^+$ satisfies $c3c(\U,\V,B{=}x_s,\dots,x_0{=}\X_{\U})$.
We distinguish two cases.

\subsubsubsection{Case 4a: \starcondition holds, i.e.,
$(\U,\X)$ and $(\W,\X)$ are edges.}\caselabel{4a}
We only have to prove Lemma~\ref{lem:T_end}  since this implies Lemma~\ref{lem:T_int}.
Consider Figs.~\ref{fig:case4a-pplus} and~\subref{fig:case4a-rep}.
	Let $\langle x_0{=}\X_{\U}{=}\U,x_1,\dots,x_s{=}\W\rangle$ be a leftmost $\W$-necklace.
	By $S_{\W\X}=(\W,\X)$, we have $t_i=\X$ for all $i$. 
	Since $x_0=\X_\U=\U$, we have that~$G_0^+$ satisfies $c3c(\U{=}x_0,x_1,\dots,x_s{=}\W,\V)$.
	But observe that $G_0^+$ has no cutting pair $\{x_i,x_j\}$ with $0\leq i < j \leq s$, 
  for otherwise the face $f$ containing $x_i$ and $x_j$ could be used as a 
  shortcut and $\calN_0$ was not leftmost (see Fig.~\ref{fig:necklace}).  
  So $G_0^+$ actually satisfies $c3c(\U,\W,\V)$.
	Use induction to obtain a \Tint-path $P_0$ from $\U$ to $\W$ in $G_0^+$
	that uses edge $(\V,\W)$.
	Then $P^+=P_0\cup (\W,\X)$ is a path in $G^+$ that contains $(\V,\W)$, and $(\W,\X)$. 

	Fix some $i=1,\dots,s$.
	If $P^+$ used edge $(x_{i-1},x_i)$ and it was virtual, then by
	Claim~\ref{cl:virtual} $f_i$ contains no vertex of $S_{\W\X}$, which
	means that the interior of $G_i$ is non-empty.
	Apply the substitution trick to remove $(x_{i-1},x_i)$ from $P^+$,
	replacing it with a path through $G_i$.    Otherwise, 
  we keep~$G_i$ as a $P^+$-bridge. We let its representative be~$x_i$ 
  if $1\leq i<s$, and $\X$ if $i=s$.
	Observe that this representative is interior or $\X$, and was not
	used by $P_0$ since~$P_0$ was a \Tint-path.  So we obtain a
	\Tend-path with the desired properties.

\subsubsubsection{Case 4b: \starcondition does not hold.}\caselabel{4b}
We must prove only Lemma~\ref{lem:T_int} and may therefore by
symmetry assume that $\U\neq \V$. 
We claim that this implies that $\deg(\X_\U)\geq 3$.  For if
$\deg(\X_\U)=2$, then its neighbors form a cutting pair, which
by corner-3-connectivity means that $\X_\U$ is a corner, hence
$\X_\U=\U$.  Since $\U\neq \V$, the two neighbors of $\X_\U$
are then $\X$ and a vertex on the left side, 
and we could have applied Case~\ref{case:2}.  Contradiction, so $\deg(\X_\U)\geq 3$.
Let $(\X_{\U},x_1)$ be the edge at $\X_\U$ that comes after $(\X_\U,\X)$
in clockwise order (see Fig.~\ref{fig:case4b-construction}).
Note that $x_1$ is face-adjacent to $\X$.  It must be
an interior vertex, for otherwise by $\deg(\X_\U)\geq 3$ edge $(\X_\U,x_1)$ is
a cutting pair that we could have used for Case~\ref{case:3} or \ref{case:3'}.

Let $\calN_0=\langle x_0{=}\X_{\U},x_1,\dots,x_s{=}\V\rangle$ be a 
simple interior $\V$-necklace; see Fig.~\ref{fig:case4b-construction}.   
We use a $\V$-necklace that is leftmost among all $\V$-necklaces that
contain $x_1$.  Note that Claim~\ref{cl:inner} holds for $\calN_0$
even with this restriction, since $(x_0,x_1)$ is an edge.
We know that $G_0^+$ satisfies $c3c(\U,\X_\U,x_1,\dots,x_s{=}\V)$.
But observe that $G_0^+$ has no cutting pair $\{x_i,x_j\}$ for $1\leq i < j \leq s$, 
for otherwise (as in Case~\ref{case:4a}) $\calN_0$ would not be the leftmost necklace
that uses $x_1$.  So $G_0^+$ actually satisfies $c3c(\U,\X_{\U},x_1,\V)$.

By induction, obtain a \Tint-path $P_0$ in $G_0$ from $\V$ to $\U$ through edge
$(x_1,x_0)$.  Append the path
$\langle \V,\W,t_s,\ldots,t_0{=}\X\rangle$ to the reverse of $P_0$ to 
obtain path $P^+$.
This path begins at~$\U$, ends at~$\X$, and contains $(\V,\W)$.
Any $P^+$-bridge is either a $P_0$-bridge (and receives a representative
there) or is $G_i$ for some $1\leq i\leq s$.  
For $i>1$, assign representative $x_{i-1}$ to~$G_i$; this is an interior vertex.  
Graph~$G_1$ has an empty interior by choice of $x_1$, so it needs no
representative.

There are two reasons why we cannot always use $P^+$ for the result.
First, it may use
virtual edges and hence not be a path in $G$.  Second, some
$P^+$-bridge $G_i$ may have four attachment points.    Both are
resolved by expanding $P^+$ via paths through the $G_i$'s.
Fix one~$i$ with $1\leq i \leq s$ and consider the following cases:

\smallskip
\noindent\textbf{Case 4b-1: $(x_{i-1},x_i)$ is virtual and
used by $P^+$, and $t_{i-1}=t_i$.}\caselabel{4b-1} 
By Claim~\ref{cl:virtual}, the interior of graph $G_i$ is non-empty
and inside the separating triplet $\{x_{i-1},x_i,t_i\}$.
Replace $(x_{i-1},x_i)$ by a path through $G_i$ with the substitution trick,
see graph $G_3$ in Fig.~\ref{fig:case4b}.

\newcommand{\casefourbCaption}{Case~\ref{case:4b}. (a) Construction of $P^+$ with representatives;
(b) Graphs 
$G_i^\asymp$ (Case~\ref{case:4b-2}) and~$G_i^\sqsubset$ (Case~\ref{case:4b-3}) used to substitute virtual edges;
(c) $P^+$ with representatives after all substitutions.}
\begin{figure}[t]
  \centering
  \begin{subfigure}[b]{.3\textwidth}
    \centering
    \includegraphics[page=4,scale=.75]{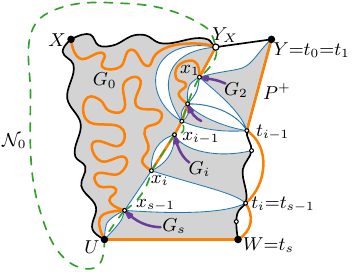}
    \caption{}
    \label{fig:case4b-construction}
  \end{subfigure}
  \hfil
  \begin{subfigure}[b]{.3\textwidth}
    \centering
    \includegraphics[page=9,scale=1]{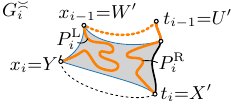}
    
    \bigskip
    \includegraphics[page=11,scale=1]{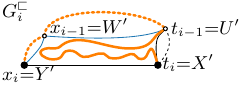}
    \caption{}
    \label{fig:case4b-substitution}
  \end{subfigure}
  \hfil
  \begin{subfigure}[b]{.3\textwidth}
    \centering
    \includegraphics[page=7,scale=.75]{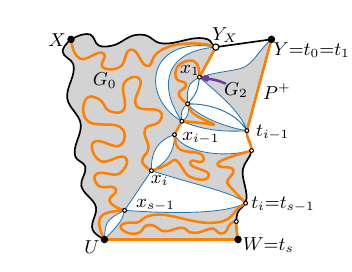}
    \caption{}
    \label{fig:case4b-result}
  \end{subfigure}
\caption{\casefourbCaption}
\label{fig:case4b}
\end{figure}

\smallskip
\noindent\textbf{Case 4b-2: $(x_{i-2},x_i)$ is virtual and
used by $P^+$, and $t_{i-1}\neq t_i$.}\caselabel{4b-2} 
See Fig.~\ref{fig:case4b-substitution}(top).
We want to replace both $(x_{i-1},x_i)$ and
$(t_{i-1},t_i)$ (which is always used by $P^+$) with a path through~$G_i$.
Let $G_i^\asymp$ be the graph $G_i$ with $(t_i,x_i)$ and $(t_{i-1},x_{i-1})$ added. 
The outer face of~$G_i^\asymp$ is a simple cycle since $f_i$ contains no
vertex of the right side by Claim~\ref{cl:virtual}, so $G_i^\asymp$
satisfies $c3c(t_i,t_{i-1},x_{i-1},x_i)$. 
By induction, find a \Tint-path $P_i$
in $G_i^\asymp$ from $t_i$ to $x_i$ that uses the edge $(t_{i-1},x_{i-1})$.
So removing $(t_{i-1},x_{i-1})$ from $P_i$ splits it into two paths: path~$P_i^{\operatorname{R}}$ 
connects $t_i$ to~$t_{i-1}$, and path $P_i^{\operatorname{L}}$ connects $x_{i-1}$ to $x_i$. 
(No other split is possible by planarity.)
Neither path uses the added edge $(t_{i},x_{i})$ since it connects the ends of $P_i$.
Use $P_i^{\operatorname{R}}$ to replace $(t_{i-1},t_i)$ 
and~$P_i^{\operatorname{L}}$ to replace $(x_{i-1},x_i)$ in~$P^+$.

\smallskip
\noindent\textbf{Case 4b-3: Subgraph $G_i$ has a non-empty interior 
and $t_i\neq t_{i-1}$.}\caselabel{4b-3}
See Fig.~\ref{fig:case4b-substitution}(bottom).
In this case, $G_i$ is a $P^+$-bridge with four attachment points, a
violation of Tutte path properties. 
If Case~\ref{case:4b-2} applied to $G_i$, then 
$G_i$ is no longer a bridge of the resulting path and we are done.  Otherwise,
we do a substitution that uses a different supergraph of $G_i$.

Let $G_i^\sqsubset$ be $G_i$ with edges from path $\langle t_{i-1},x_{i-1},x_i,t_i\rangle$
added if not already in $G_i$.  This graph satisfies $c3c(t_i,t_{i-1},x_{i-1},x_i)$ 
and satisfies condition \starcondition if we set $\U'=t_i$, $\V'=t_{i-1},
\W'=x_{i-1}$ and $\X'=x_i$.
So we can find a \Tend-path $P_i'$ of $G_i^\sqsubset$ from $t_{i}$ to $x_{i}$ that uses
$(t_{i-1},x_{i-1})$ and $(x_{i-1},x_i)$.
Thus,~$P_i'$ ends with $\langle t_{i-1},x_{i-1},x_{i}\rangle$ 
and $P_i'\setminus \{(t_{i-1},x_{i-1}),(x_{i-1},x_{i})\}$
is a path from $t_{i-1}$ to $t_i$ in $G_i$ that does not visit $x_{i-1}$ or $x_i$.
Substitute this path in place of edge $(t_{i-1},t_i)$ in~$P^+$.  
Note that one $P_i'$-bridge $C$ may use $x_{i}$ as its representative, but 
if so, then it also has $x_{i-1}$ as attachment point.    We set $x_{i-1}$
(which was $G_i$'s representative and is no longer needed as such) to be
the representative of $C$.

\smallskip
\noindent\textbf{Case 4b-4: $t_{i-1}\neq t_i$ and $(t_{i-1},t_i)$ is virtual}.\caselabel{4b-4} 
Since $P^+$ always uses edge $(t_{i-1},t_i)$, we must replace this edge
with a path through $G_i$.  We claim that 
this is done automatically
because Case~\ref{case:4b-3} applies. 
Namely, if $(t_{i-1},t_i)$ is virtual, then there is
at least one vertex between $t_{i-1}$ and $t_i$ on the right side.
This vertex is exterior in $G$ and hence neither $x_i$ nor~$x_{i-1}$.
So it is strictly inside $\calN_i$, hence $G_i$ has a 
non-empty interior and (by $t_{i-1}\neq t_i$) Case 4b-3 applies.

\medskip
After doing these substitutions, there are no virtual edges in the path, 
no bridges have four attachment points, 
every bridge has an interior vertex as representative,
and no vertex was used twice as representative; see Fig.~\ref{fig:case4b-result}.
This ends the proof of Lemma~\ref{lem:T_int} and \ref{lem:T_end}.

\subsection{Linear time complexity}

It should be clear that our proof is algorithmic.
The main bottlenecks for its running time are to determine which case to apply (i.e., whether there is
a cutting pair) and to find the $B$-necklace.  Both can
be done in linear time, by computing all cutting pairs~\cite{DBT89,GM00}
and by finding a leftmost path in the subgraph induced
by vertices that are face-adjacent to $S_{\W\X}\setminus B$.  This would yield quadratic running time overall.
For triangulated planar graphs, this is easily reduced to linear:
cutting pairs correspond to interior edges where both ends are exterior, 
and the necklace can be found, as in  \cite{AKS84}, with a left-first search that only advances 
neighbours of $S_{\W\X}\setminus B$.  But for graphs that are
not triangulated we need a few extra data structures.  We sketch only some
ideas for this here; details are in Section~\ref{sec:runningtime}. 

Globally, we keep track of the corners $\U,\V,\W$, and~$\X$.
For each \emph{interior vertex $w$} and every \emph{side}~$S_{ab}$, we keep a list $\calV(w,S_{ab})$
of faces that contain $w$ as well as a vertex on $S_{ab}$.
In these lists, we can look up quickly whether an interior vertex is face-adjacent to a side.
Also, each face knows for each side which vertices it has on it.
Finally, for each \emph{pair of sides} $S_{ab}$ and $S_{cd}$, we store
a list $\calP(S_{ab};S_{cd})$ of faces that are incident to a vertex on~$S_{ab}$ and a (different) vertex 
on~$S_{cd}$, i.e., faces that connect cutting pairs. 

This allows to test for Case~\ref{case:2} and Case~\ref{case:3} easily (``is $\calP(S_{\U\V},S_{\W\X})$
resp.~$\calP(S_{\W\X},S_{\X\U})$ non-empty?''), and Case~\ref{case:1} and 
Case~\ref{case:4} are easily determined
from the planar embedding.  We keep $\calP(S_{\W\X},S_{\X\U})$ in an order such that its first entry
is the appropriate cutting pair in Case~\ref{case:3}.  To find a necklace, we {\em scan} 
the faces incident to $x_1,\dots,x_s$.  More precisely, we consider
(for vertex $x_i$, presuming we know face $f_i$ already) each face $f$ in ccw order after $f_i$,
and along face $f$ each vertex $w$ in ccw order after $x_i$, until we find vertex~$B$ 
(then we are done) or a vertex that is face-adjacent to a vertex
in $S_{\W\X}\setminus B$  (then this is $x_{i+1}$
and $f_{i+1}=f$ and we repeat).  The running time for this is proportional to the degrees of
vertices and faces that were scanned.    We also need to update the data structures when
recursing into a subgraph; here, we scan along all vertices (and their incident faces) that
were in some necklace along which we cut the graph, or that became newly exterior.

Two crucial insights are needed to bound the running time.  First,
we need to scan vertices and faces only if they become incident to a side that they were
not previously incident to.  Second, once a vertex or
face is incident to a side, it remains incident to it forever (for some suitable definition
of ``side'').  The two combined mean that every vertex and face is scanned only a constant
number of times, because there are only four sides to have incidences with, and the linear
running time follows.  In fact, we only scan vertices and faces that are incident to the
outer face in some subgraph, which means that they will be incident to the path $P$
that we compute, and we have the following:

\begin{theorem}
\label{thm:linearTime}
\label{thm:time_complexity}
The Tutte path $P$ for Theorem~\ref{thm:2conn}, Lemma~\ref{lem:T_int} or 
Lemma~\ref{lem:T_end} can be found in linear time.  
More specifically,
the running time is $O(\sum_{f\in F(P)} \deg(f))$.
\end{theorem}

%%%%%%%%%%%%%%%%%%%%%%%%%%%%%%%%%%%%%%%%%%%%%%%%%%%%%%%%%%%%%%%%%%%%%%%%
\section{Applications}

A number of interesting properties of planar 3-connected graphs can
be derived easily from the existence of $\TSDR$-paths.  In particular,
every planar 3-connected graph has a spanning tree of maximum degree~3
\cite{Bar66}
(a concept known in the literature as a {\em 3-tree}, but we prefer
to use the term {\em binary spanning tree} to avoid confusion with
maximal graphs of treewidth~3).  Secondly, every planar 3-connected
graph has a {\em 2-walk}, i.e., a walk that visits every vertex
at least once and at most twice \cite{GR94}.  In Section~\ref{sec:3trees} 
and Section~\ref{sec:2walks}, 
we show that, using Lemma~\ref{lem:T_int}, these can be found in linear
time; this was known for binary spanning trees 
\cite{Strothmann-PhD,Bie-Barnette}, but for 2-walks the previous
best running time was $O(n^3)$ \cite{SS-ICALP18}. 

\begin{theorem}
\label{thm:3tree}
Let $G$ be a 3-connected plane graph with exterior vertex $\U$. 
Then $G$ has a binary spanning tree $T$ that can be found in linear time. 
Moreover, when rooting $T$ at $\U$,
a vertex $v$ has two children only if it is an interior vertex
	and part of a cutting triplet $\{v,w,x\}$ of $G$;
	one of the subtrees of $v$ contains exactly the
	vertices interior to $\{v,w,x\}$.\
\end{theorem}

\begin{theorem}
\label{thm:2walk}
Let $G$ be a 3-connected plane graph with exterior vertex $\U$.
Then $G$ has a 2-walk $P$ that can be found in linear time.
Moreover, $P$ visits $\U$ exactly once, and it visits
a vertex $v$ twice only if $v$ is part of a separating triplet.
\end{theorem}

\section{Additional proof details}\label{sec:additional}

Throughout the text, we have left some of the more obvious details
to the reader.  This section provides detailed proofs for some of
these results.

\subsection{Arguing corner-3-connectivity}

We defined corner-3-connectivity via the corner stellation,
but to show it for some subgraph~$G'$, we used one of the following
two arguments:  Either $G'$ was defined as the graph inside some noose~$\calN$, 
and the outer face of $G'$ was a simple cycle, or we already knew that $G'$
was corner-3-connected for some corners, but we claimed that we can omit
some of the corners since there are no cutting pairs near them.  We now
formally prove that this is correct.

\begin{lemma}
\label{cl:noose}
Let $G$ be a graph that satisfies $c3c(\U,\V,\W,\X)$ for some corners
$\U,\V,\W,\X$.
Let $H$ be the graph inside a noose $\calN=\langle x_1,\ldots,x_k\rangle$ that visits the outer face at most once.
Let $H^+$ be the graph obtained from $H$
by adding some edges along~$\calN$.
If the outer face of~$H^+$ is a simple cycle, then $H^+$ satisfies $c3c(\{x_1,\dots,x_k\}\cup  \calC)$,
where $\calC \subseteq \{\U,\V,\W,\X\}$ are the corners inside \calN.  
\end{lemma}
\begin{proof}
See Fig.~\ref{fig:claim-noose} for an illustration of this proof.
Modify the noose \calN such that it also goes through the vertices in
$\calC$ while going through the outer face.  Now the noose \calN contains
vertices $\langle x_1,\dots,x_\ell \rangle$, where $\calC=\{x_{k+1},\dots,x_\ell\}$.

\begin{figure}[t]
\centering
\subcaptionbox{$H$}{\includegraphics[scale=.75,page=1,trim={1cm 0 .7cm 0},clip]{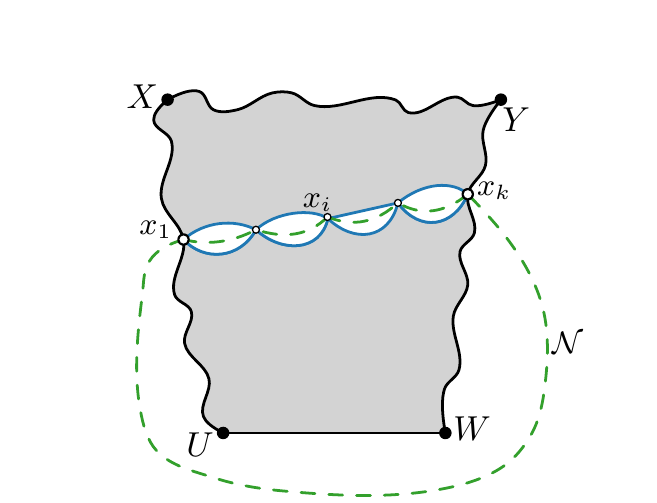}}
\hfil
\subcaptionbox{$H^+$}{\includegraphics[scale=.75,page=2,trim={1cm 0 .7cm 0},clip]{noose}}
\hfil
\subcaptionbox{$\stel H$}{\includegraphics[scale=.75,page=4]{noose}}
\caption{Illustration for Lemma~\ref{cl:noose}.}
\label{fig:claim-noose}
\end{figure}

Consider the corner stellation $\stel{H}$ of $H^+$ with respect to $\{x_1,\dots,x_\ell\}$; 
with our modification, noose \calN is a noose in $\stel{H}$ as well.
Observe that any pair $x_i,x_j$ is connected by three vertex-disjoint paths in $\stel{H}$:
two along the outer face of $H^+$ 
and one via the stellation vertex $s$.  Likewise, $s$ has three vertex-disjoint
paths to any $x_i$ in $\stel{H}$: Use edges $(s,x_{i-1}),(s,x_i),(s,x_{i+1})$, and
the outer face cycle of $H^+$.  So if $\stel{H}$ had a cutting pair $\{v,w\}$, then
all of $s,x_1,\dots,x_\ell$ are in one cut component of $\{v,w\}$.  In particular,
therefore $v\neq s\neq w$ since $s$ has no other neighbors.  Let $C$ be some
cut component of $\{v,w\}$ that does not contain~$s$.  Then $C-\{v,w\}$ contains
no vertices in $s,x_1,\dots,x_\ell$, so all of them are strictly inside~$\mathcal N$
and have no neighbor (in $G$) outside $\mathcal N$.  Therefore, $\{v,w\}$ is
also a cutting pair of~$\stel{G}$, a contradiction.
\end{proof}

\begin{lemma}
\label{lem:skip_corner}
Assume that $G$ satisfies $c3c(c_1,\dots,c_k)$ for some $k\geq 4$.  
Assume that no cutting
pair $\{v,w\}$ of $G$ satisfies $v,w \in S_{c_1c_3}$.
Then $G$ also satisfies $c3c(c_1,c_3,\dots,c_k)$.
\end{lemma}
\begin{proof}
Let $\stel{G}$ be the corner stellation with respect to all
corners $c_1,\dots,c_k$.  To prove the claim, we must show that
$\stel{G}\setminus (c_2,s)$ (where $s$ is the stellation vertex)
is also 3-connected.  Assume that it is not, in case of which it has
a cutting pair $\{v,w\}$ that is necessarily also a cutting pair of $G$.
Let $C_1,C_2$ be two cut components of $\{v,w\}$ in $G$ (there are only
two since~$G$ is internally 3-connected).  Since $\stel{G}$
is 3-connected, both $C_1^+$ and $C_2^+$ contain a corner other than $v,w$.
This implies that both $v$ and~$w$ lie 
on the outer face of $G$.  Since they are not both in $S_{c_1c_3}$, therefore
both outer face paths between $v$ and $w$ (and hence both $C_1^+$ and~$C_2^+$)
contain at least one corner other than $c_2,v,w$.    Therefore, $C_1^+$ and
$C_2^+$ are connected in $\stel{G}\setminus (c_2,s)$ by going from those
corners to $s$.  So $\{v,w\}$ is not a cutting pair of $\stel{G}\setminus
(c_2,s)$, a contradiction.
\end{proof}

\subsection{Combining \texorpdfstring{\Tint}{T\_int}-paths}

We generally obtained a \Tint-path by combining two or more \Tint-paths of
subgraphs, possibly omitting some edges from these paths (but in such a way
that the combination is a simple path).  We now formally verify that this
indeed gives a \Tint-path.

\begin{lemma}
\label{cl:combine}
Let \calN be a noose of $G$
and let $H_i$ and $H_o$ be the graphs inside and outside \calN, respectively.
Let $P_i$ and $P_o$ be \Tint-paths of $H_i$ and $H_o$, respectively.
Let $P\subseteq P_i\cup P_o$ be a simple path that visits all vertices of $P_i\cup P_0$.
Then $P$ is a \Tint-path of $G$.
\end{lemma}
\begin{proof}
Since $P_i$ and $P_o$ visit all their respective exterior vertices,
$P$ visits all exterior vertices and all vertices in \calN.

Consider a $P$-bridge $C$.  It cannot have vertices both strictly inside
and strictly outside~\calN, since all vertices on \calN are visited by
$P$ and $C$ is connected.  So $C$ is a $P$-bridge of $H_i$ or~$H_o$,
and can inherit the representative it received from there.  No representative
is used twice since the representatives of $P_i$ and $P_o$ were strictly
inside/outside \calN, respectively.  Also, no representative is on the
outer face, since none of $P_i$ and $P_o$ were.
\end{proof}

\subsection{Existence of necklaces}

We must argue that a suitable necklace exists in Case~\ref{case:4}.

\begin{lemma}
If none of Case~\ref{case:1},\ref{case:2},\ref{case:3},\ref{case:3'} applies, then $G$ has a simple interior
$B$-necklace for $B\in \{\V,\W\}$.
\end{lemma}
\begin{proof}
It is easy to see that a $B$-necklace exists; for example,
we can take all neighbors of $S_{\W\X}\setminus \{B\}$
that are not on the right side, enumerating  them 
in order from $\X$ towards $\W$ and in ccw order
at each vertex.  We claim that the resulting 
necklace $\X_{\U}{=}x_0,\dots,x_s{=}B$
is interior. For if $x_i$ (for some $0 < i < s$) were exterior,
then $x_i$ and its neighbor $t_i\in S_{\W\X}$ would form a cutting pair
separating $x_0$ and $x_s$.
So if $x_i$ is on the left side, then we are in Case~\ref{case:2},
and if it is on the top side, then we are in Case~\ref{case:3}
(which can be applied since $(x_i,t_i)$ is an edge).

Consider the curve $\calC$ defined by this necklace.
If $\calC$ visits a vertex twice, then shortcut the necklace by
removing the part between the two visits.   If $\calC$ crosses 
itself, say $x_{i-1}{-}x_i$ intersects $x_{j-1}{-}x_j$ for $i<j$,
then we can immediately go from $x_{i-1}$ to $x_j$ with a curve
along~$\calC$, thereby removing $x_i,\dots,x_{j-1}$ from the necklace.
So the shortest possible necklace that uses only vertices in
$\{x_0,\dots,x_s\}$ is simple and interior.
\end{proof}

\section{Linear-time complexity for 3-connected graphs}\label{sec:runningtime}

We now explain how to find the $\Tint$-path in linear time
for 3-connected graphs, i.e., prove Theorem~\ref{thm:time_complexity}.
We first argue how to bound the time spend on recursions, once we know
which case applies and have found the subgraphs.  (This is the easier
part.)  Next, we explain how to store cutting pairs (for determination of
cases) and how to store adjacencies to sides (for determination of the
necklace); these data structures are not complicated, but arguing that their
updates take overall linear time is lengthy.

\subsection{Preliminaries.}
A few notations will be useful.  First, set $D_F:=\sum_{f\in F(P)} \deg(f)$;
we aim to show that the running time is $O(D_F)$.  Next,
set $D_V:= \sum_{v\in V(P)} \deg(v)$,
where $V(P)$ are the vertices of~$P$.    Since every vertex has an incident
interior face, we have $D_V\leq D_F$; it hence suffices to argue a running time
of $O(D_V+D_F)$.  Next,
set $\calG$ to be the set of all subgraphs that we used in some recursion.
Let $G^+$ be the graph that we would get if we inserted into $G$ all
the edges that were used as virtual edge in some recursion.  Note that
$G^+$ is still a planar graph.  

Let $V_X$ be the set of vertices of $G$ that were exterior in some subgraph
$G' \in \calG$.  Note that, when recursing into $G'$, we obtained a Tutte
path of $G'$ that visits all exterior vertices in $G'$.  When combining
Tutte paths of subgraphs, the resulting path always visits the same set
of vertices.
So $V_X\subseteq V(P)$, and in particular $|V_X| \leq |V(P)| \leq D_V$.

Let $E_X$ be the set of edges of $G^+$ (i.e., possibly including some virtual
edges) that were exterior in some subgraph $G'\in \calG$.  The ends of such
edges necessarily belong to $V_X$.  Since~$G^+$ is planar, we have
$|E_X|\leq 3|V_X|-6 \in O(D_V)$.

For any subgraph $G'\in \calG$, let $E_X(G')$ be the edges of $G^+$ that
are exterior in $G'$, but were not exterior (or did not exist) in the 
parent graph that we recursed from.  Note that $|E_X(G')|\geq 1$ in all
cases.  The work done in the recursion for $G'$ (not including
the time to determine the case or to find the necklace) is $O(|E_X(G')|)$, because
we must update the planar graph embedding at the places where
we split the parent graph to obtain $G'$ (details are given below).    
Observe that $\sum_{G'\in \calG} |E_X(G')| \leq  2|E_X|$, because any
edge becomes exterior only once, and then belongs to at most two
subgraphs that we recurse in.  Therefore, the time to handle each
recursion (excluding the time to find the case and the necklace) is
$O(D_V)$.  The bottleneck for proving Theorem~\ref{thm:linearTime}
is hence to show how we can test for cutting pairs and find leftmost
necklaces in overall time $O(D_V + D_F)$.

\subsection{Data structures.}
We presume that with the planar embedding we obtain a standard data structure with the following:
\begin{itemize}
\item Every vertex $v$ knows whether it is exterior, and has a list $L(v)$ of its incident faces and edges in ccw order.  
This list is circular if $v$ is interior, and begins and ends with the exterior edges at $v$ otherwise.
(There are exactly two such edges since all graphs in $\calG$ are 2-connected.)
\item Every interior face $f$ has a list $L(F)$ of its incident vertices and edges in ccw order.  

One could demand a similar list for the outer face, but since our outer face changes frequently we will
not do this.  Observe that even without such a list we can walk along the outer face in ccw order, presuming we
know at least one vertex $v$ on it, by traversing the last edge in~$L(v)$ and recursing from its other end.
\item Vertices, edges and faces are cross-linked, i.e., each vertex/edge/face knows all its occur\-ren\-ces in the lists $L(v)$ and/or $L(F)$.
\end{itemize}
Setting up this data structure is standard material and will not be explained here; see, for example, \cite{C97,MP78}.
To find cutting pairs and leftmost necklaces efficiently, we need to 
store more information as follows; see Fig.~\ref{fig:datastructures}.
\begin{itemize}
\item[\calC:] We store a circular list $\calC$ of 
four references $c_0,c_1,c_2,c_3$ to vertices in $G$; these are the corners of $G$ enumerated in ccw order.  We also know which of these corners is $\U$.
Note that the ``start corner'' $\U$ may change when recursing in a subgraph 
(consider for example graph~$G_0^+$ in Case~\ref{case:4b}), but we keep the vertex in 
the same place in $\calC$ and only change the reference to~$\U$; see Fig.~\ref{fig:datastructures-C}
\end{itemize}

\begin{figure}[t]
\centering
\subcaptionbox{$\mathcal C$ and $S_i$\label{fig:datastructures-C}}{\includegraphics[page=1,scale=.65]{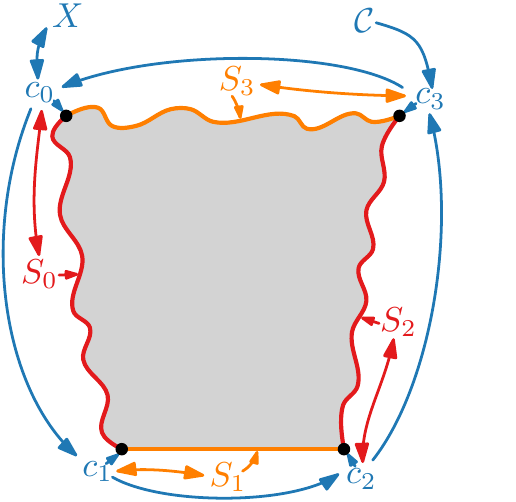}}
\hfil
\subcaptionbox{$\mathcal F(f,i)$\label{fig:datastructures-F}}{\includegraphics[page=2,scale=.65]{datastructures}}
\hfil
\subcaptionbox{$\mathcal V(w,i)$\label{fig:datastructures-V}}{\includegraphics[page=3,scale=.65]{datastructures}}
\hfil
\subcaptionbox{$\mathcal P(i,j)$\label{fig:datastructures-P}}{\includegraphics[page=4,scale=.65]{datastructures}}
\caption{Illustration of the data structures.}
\label{fig:datastructures}
\end{figure}

We do not explicitly store the sides; any side is uniquely determined by its corners.
We use $S_i$ to denote the side from $c_i$ to $c_{i+1}$
(addition for corners is always modulo 4).  
It is important that sides are defined via entries in $\calC$ ({\em not} via references to vertices) due to the following.  When recursing in a subgraph
$G'$, we may replace an entry $c_i$ in $\calC$ by a new vertex $c_i'$;
we call $c_i'$ the {\em corner corresponding to $c_i$} .   With
this, $S_i$ now automatically refers to the side that starts at $c_i'$ in $G'$
(we call this the {\em corresponding side in $G'$}).  

We call an interior face $f$
{\em incident to side $S_i$} if it contains a
vertex $v$ on $S_i$.  We order the interior faces incident to side $S_i$
as follows:  We start with the interior faces around $c_i$ in clockwise order,
hence ending with the interior face incident to $c_i$ and its neighbor~$u_1$ on~$S_i$.
Next come the remaining interior faces around~$u_1$, in clockwise order, ending with the interior face incident
to $u_1$ and its neighbor $u_2\neq c_i$ on~$S_i$.  Continue until we reach
vertex $c_{i+1}$.  With this definition in place, we need the following three lists
at vertices, faces, and sides:
\begin{itemize}
\item[\calF:] For every interior face $f$ and every $i=0,\dots,3$, we store a (possibly empty) list $\calF(f,i)$ of vertices that are on $f$ and on side $S_i$.
There are at most two such vertices per side by corner-3-connectivity, else 
the two non-consecutive vertices of $f$ on $S_i$ would form a cutting pair
within $S_i$; see Fig.~\ref{fig:datastructures-F}.

This obvious fact is the crucial ingredient for our data structures, because 
with this many checks can be done in constant time
that otherwise would take longer.  For example, given a vertex $w$, 
we can check in constant time which side(s) (if any)
it lies on.  Namely, take an arbitrary interior face $f$ at $w$ and
inspect $\calF(f,i)$ for each $i=0,\dots,3$.  Since each list has constant size, we can
check in constant time whether $w$ is in it, hence whether it lies on $S_i$.

\item[\calV:] For every interior vertex $w$ and every $i=0,\dots,3$, we store a (possibly empty) list $\calV(w,i)$ of interior faces $f$ that contain
$w$ as well as a vertex on side $S_i$.
The list is sorted by the order in which these faces are incident to side $S_i$; see Fig.~\ref{fig:datastructures-V}.

Note that $\calV(w,i)$ allows us to check in constant time
whether $w$ is face-adjacent to corner~$c_i$, 
because any interior face that contains $w$ and $c_i$ would have to be
the first entry in $\calV(w,i)$ by the order of faces incident to $S_i$.

\item[\calP:] For any two distinct sides $S_i,S_j$
we store a list $\calP(i,j)$ of interior faces
that contain a cutting pair~$\{v,w\}$ of $G$, with $v$ on side $S_i$
and $w$ on side $S_j$.
This list is sorted by the order in which these faces are incident to side $S_i$; see Fig.~\ref{fig:datastructures-P}.

Note that $\calP(j,i)$ is the reverse of $\calP(i,j)$, but it will
be convenient to store both of them.  We assume
that any face $f$ knows of all its occurrences in some list $\calP(i,j)$;
this is only a constant overhead per face. 
\end{itemize}

\subsubsubsection{Initialization.} 
We argue how to initialize the data structure in $O(n)$ time; this then also
shows that it uses $O(n)$ space.    We assume that we are given references
to the corners,
and hence initialize~$\calC$ by walking around the outer face.  
For all other data
structures, we call \scanside{$i,c_i,c_{i+1}$} for $i=0,\dots,3$.
See Algorithms~\ref{alg:scanside} and \ref{alg:scanface}.
This visits all exterior vertices, scans all their
incident interior faces, and updates all the lists as needed.  One verifies
that, since we scan along a side in ccw order, the lists automatically
receive the correct order.  We should also mention that the conditions in lines
\refline{+} and \refline{++} (explained below) are always true for the initialization 
and can be tested in constant time for later recursions.
Condition \refline{+++} can be tested in constant time per vertex
by marking (before starting any scan for any side)
all vertices that will be scanned; in the initialization, this is
all vertices.

\begin{algorithm}[t]  
  \caption{\protect\scanside{$i,c,d$}} 
\label{alg:scanside}
  
  \Input{$c$ and $d$ are vertices on the outer face, the ccw path from~$c$ to~$d$ is on side $S_i$}
  \BlankLine
  \ForEach{vertex $v$ on side $S_i$, from $c$ to $d$}{
    \If(\tcp*[f]{\textlabel{(+)}{line:+}}){$v$ was not on side $S_i$ before}{
      \ForEach{interior face $f$ in $L(v)$ in clockwise order after the outer face)}{
        add $v$ to $F(f,i)$ \;
        \If(\tcp*[f]{\textlabel{(++)}{line:++}}){$f$ was not incident to side $S_i$ before}{
          \scanface{$f,v,i$}
        }
      }
    }
  } 
\end{algorithm}

\begin{algorithm}[t]
  \caption{\protect\scanface{$f,v,i$}}
\label{alg:scanface}
  
  \Input{$f$ is an interior face and incident to exterior vertex $v$ on side $S_i$}
  \BlankLine
  \ForEach{vertex $w \in L(F)\setminus\{v\}$ in clockwise order after $v$}{
    \If{$w$ is interior}{
      append $f$ to $V(w,i)$
    }
    \ElseIf{edge $(v,w)$ is not on the outer face}{
      \ForEach{side $S_j\neq S_i$ containing $w$}{
        append $f$ to $P(i,j)$ \;
        \If(\tcp*[f]{\textlabel{(+++)}{line:+++}}){$w$ will not get scanned}{
          prepend $f$ to $P(j,i)$ \;
        }
      }
    }
  } 
\end{algorithm}

The operation of scanning a vertex $v$ takes $O(\deg(v))$ time
plus the time to scan its incident interior faces which we will count
with those faces.
The operation of scanning a face $f$ takes $O(\deg(f))$ time.
In total, the running time of the initialization is therefore
proportional to the degrees of vertices and faces that are incident to the
outer face, which is $\Theta(D_V+D_F)$ since $P$ visits all exterior
vertices.

\subsection{Updating}
We first give some general rules for how to update the data structure,
and then fill in for each individual case some case-dependent details.
Assume that we recurse into some subgraph $G'$ of~$G$.   Where
useful, we will use ``primed'' versions (such as $\calV'$ and $S'_i$)
for the data structures and properties of $G'$.

Each case will state which
corner $Z'$ of $G'$ {\em corresponds} to which corner $Z$ of $G$, i.e., 
takes the entry of $Z$ in $\calC$;  this defines corresponding 
sides.  We will do this 
so that the following holds.

\begin{property}
\label{prop:vertex_side}
Assume that a vertex $v$ is incident to a side $S_i$ in $G$ 
and belongs to some subgraph~$G'$ that we recurse in.
Then $v$ is incident to the corresponding side $S_i'$ of $G'$.
\end{property}

We cannot afford to fully copy the graph structure (i.e., lists $L(v)$ and
$L(f)$) from $G$ into $G'$ since this would be too
slow.  Instead, observe that $G'$ is in all cases 
defined as the inside of some noose $\calN$ of $G$.
At every vertex $v$ traversed by noose $\calN$, we split the
incidence list $L(v)$ into the two parts at the faces traversed by $\calN$;
these faces become the outer face of $G'$ and hence this sets the
lists up correctly.  For any interior face $f$
traversed by $\calN$, we might insert a virtual edge along~$f$; if we do so,
then we assume that the part inside $\calN$ (hence belonging to $G'$) 
inherits the reference $f$ and the (suitably shortened) list $L(F)$.
(We assume that here, as in many of the other operations below, we store
where the lists were cut, and keep a reference to the rest, so that we can
restore the rest
when returning from the recursion in $G'$ and recursing into a different 
subgraph of $G$.  We view this record as ``belonging'' to the (possibly
virtual) edge that newly became an exterior edge of~$G'$; the final set
of such edges form a planar graph, thus this overhead takes space $O(n)$
overall.) 

Subgraph $G'$ {\em inherits} the lists $\calV$, i.e., 
$\calV'(w,i):=\calV(w,i)$ for all vertices $w$ that are in $G'$, 
and no time is spent on creating these since we simply keep the same lists.
We need to
argue that this does not create false positives for an interior vertex
$w$ of $G'$.  Recall that
$\calV(w,i)$ stores interior faces~$f$ that contain $w$ as well as a vertex
on $S_i$.  Clearly, $f$ still contains~$w$ in~$G'$, and it does
contain a vertex on the corresponding side $S_i'$ of $G'$ because of
the following property of our cases:

\begin{property}
\label{prop:sides}
\label{prop:face_sides}
Assume that an interior face $f$ is incident to side $S_i$ in $G$, 
and $f$ (or some part of $f$ obtained by dividing $f$ along a virtual edge) 
is an interior face of some subgraph~$G'$ that we recurse in.
Then~$f$ is incident to the corresponding side $S_i'$ of $G'$.
\end{property}

This property is quite obvious for a face $f$ that was not
divided, because then all vertices of~$f$ must also belong to $G'$,
and by Property~\ref{prop:vertex_side} such a vertex remains on all
sides that it was on in $G$.
This property is not at all obvious for
faces that were divided by insertion of virtual edges; we will
argue this for each case below.

Subgraph $G'$ also inherits
the four lists $\calF(f,i)$ for any interior face $f$.  If $f$ was not
divided by a virtual edge, then this does not add false
positives since any vertex that $f$ had on side $S_i$
of $G$ also is in~$G'$ (else $f$ would have been divided) and on
side $S_i'$ (by Property~\ref{prop:vertex_side}).  If~$f$
was divided by a virtual edge~$(v,w)$, then we must change 
$\calF(f,i)$, where $S_i$ is the side that contains $(v,w)$.
We know $i$ from each case below, and set $\calF(f,i):=\{v,w\}$.
(We should also keep a copy of the prior list $\calF(f,i)$ until we are done
with updating the entire data structure, because in some of the
tests below we need to use the lists as they were in $G$.  This
is only constant space overhead.)

The initialization of the $\calP$-lists for $G'$ will depend on each
case, but as will be seen, they are either empty or inherited from $G$
(with minor modifications). Furthermore, all cutting pairs of $G$ that also
exist in $G'$ get copied over.

\subsubsubsection{Scanning (parts of) sides.}
We have initialized the data structures of $G'$ so that it has no
false positives (entries that should not be there), but it may be
missing some entries since some vertices may be new to a side.
To add these missing entries, we scan (parts of) each side $S'_i$ in such a way that all 
vertices that are new to $S'_i$ are guaranteed to be scanned, i.e.,
we call \scanside{$i,c,d$} for some $i\in \{0,\dots,3\}$ and
some vertices $c,d\in S'_i$.  
(Each case listed below will explain exactly what needs to be scanned.) 
We should mention here that the pseudocode needs
a few minor modifications to obtain the correct order for $\calV$ and
$\calP$; we will explain this below.
To make this feasible, we need that $c,d$ are not arbitrary, and 
that the following holds in our cases:

\begin{property}
\label{prop:scan_side}
Assume that, when recursing from $G$ to a subgraph $G'$, we scan along
a part~$S'_{cd}$ of some side $S'_i$.  Then, one of the following holds:
\begin{enumerate}[label=(\roman*)]
\item We {\em scan the entire side} (i.e., $c=c'_i$ and $d=c'_{i+1}$), 
	and the vertices between $c$ and $d$ (if any)
	were not on side $S_i$, or
\item $d=c'_{i+1}$, $(c,d)$ is an edge, $c$ was on $S_i$ in $G$, and $d$ 
was not on $S_i$ in $G$.
\end{enumerate}
\end{property} 

It is worth pointing out that the second situation is needed in only one
case (for graph~$G_b^+$ in Case~\ref{case:3}).  

Recall Algorithm~\ref{alg:scanside} for scanning sides.
We exempt in line \refline{+} a vertex $v$ from scanning if it was already 
on side $S_i$ before.  (By Property~\ref{prop:scan_side}, this can 
happen only to $c$ or to $d$, because all other vertices were new to the side.)
First, note that this can be
tested in constant time by inspecting $\calF(f,i)$ (using the
lists as they were in $G$) for some interior face $f$ incident to~$v$.  
If $v$ was already on side $S_i$ in $G$, then scanning it
would not add $v$ or an interior face~$f$ incident to $v$
to any $\calV$-list or $\calF$-list,
because all these lists already contained it in $G$ and were inherited.
Also, any cutting pair $\{v,w\}$ of $G'$ where both $v$ and~$w$
already existed on their sides in $G$ was in some $\calP$-list
of $G$ and was inherited by $G'$.  If $w$ did not exist on its
side (say $S_j$) in $G$, 
then we will scan at $w$ from side $S_j'$, find the cutting pair then,
and update $\calP(i,j)$ as well due to line \refline{+++}.
So there is no need to scan at $v$.

We likewise exempt in line \refline{++}
a face $f$ from scanning if it was already on the 
side, i.e., if the list $\calF(f,i)$ in~$G$ was non-empty.    
As above, one argues that this
will mean no missed entries in the data structure.
Avoiding these face scans is {\em the} crucial insight
to bring the running time down to linear.

\begin{lemma}
The total time to scan sides of subgraphs of $G$ is
$O(D_V+D_F)$.
\end{lemma}
\begin{proof}
When scanning sides (for one subgraph $G'$), we may spend time 
on vertices that end up not being scanned, due to line \refline{+}.  
By Property~\ref{prop:scan_side}, there are at most
two vertices on each side, hence~$O(1)$ in total.  We count
this as overhead to the time $O(|E_X(G')|)$ that we budgeted
for $G'$ earlier, and do not consider it further here.

Likewise, we spend some time on interior faces that
were incident to a scanned vertex~$v$, but already incident to the side of~$v$
and therefore end up not being scanned due to line \refline{++}.
We spend $O(1)$ time per such face, hence $O(\deg(v))$ time per
scanned vertex $v$. We count this as
overhead to the time that we budget for scanning $v$,
and do not consider it further here.

So we must only bound the time spend on scanning faces and vertices
that were actually new to the side.
Every vertex is incident to at most 4 sides, and since it never loses
a side incidence by Property~\ref{prop:vertex_side} and gains one
with every scan by the rule in line \refline{+}, it is scanned $O(1)$
times.

Consider a face $f$ that gets scanned at least once.  During some 
recursion, face $f$ may get split into pieces, which in turn can get split
into more pieces in later recursions; let $f_1,\dots,f_k$ be the (disjoint)
pieces of $f$ that do not get split further, and note that $k\leq \deg(f)-2$
since at the worst $f$ gets split into $\deg(f)-2$ triangles.

Consider some part $f'$ of $f$ (possibly $f'=f$) that gets scanned,
which takes $O(\deg(f'))$ time.   Let $f_{i_1},\dots,f_{i_\ell}$
(with $i_j\in \{1,\dots,k\}$) be the pieces of $f$ that belong to $f'$. 
Therefore $\sum_{j=1}^\ell \deg(f_{i_j}) \geq \deg(f')$. So it suffices
to account for the work if we assign $O(\deg(f_{i_j}))$ work to each
piece $f_{i_j}$.
During the scan of $f'$, it acquires a new side incidence due to
the rule in line \refline{++}.  Crucially, by Property~\ref{prop:face_sides},
this side incidence also exists in each $f_{i_j}$.
Therefore, we count $O(\deg(f_{i_j}))$ time
only if $f_{i_j}$ acquires a new side incidence, which happens at most 4
times.
In consequence, the total time spent on scanning
all the pieces of $f$ is $O(\sum_{i=1}^k \deg(f_i))=O(\deg(f)+k)=O(\deg(f))$.

If a vertex or face gets scanned, then it either was incident to the
outer face already (as is the case in the initialization step), or it
became incident to one more side, hence incident to the outer face. So
it is in $V_X$ (in case of a vertex) or incident to a vertex in $V_X$ (in case 
of a face).  As before, $V_X\subseteq V(P)$, so only vertices on $P$ or
faces incident to $P$ get scanned and the running time is as desired.
\end{proof}

\subsubsubsection{Maintaining the correct order.}
During \verb+side_scan(i,c,d)+, 
we possibly add entries to $\calF(f,i)$, $\calV(w,i)$, and $\calP(i,j)$,
for some interior face $f$, vertex $w$, and side $S_j$. The $\calF$-lists have
constant size and no particular order, but the other two lists must keep 
entries in the order in which faces are incident to side $S_i$. 
Property~\ref{prop:scan_side} is crucial for showing this: we usually scan 
nearly the whole side.

Assume first that we scan an edge $(c,c'_{i+1})$, where $c'_{i+1}$ is new 
to side $S'_i$.
Therefore, the existing entries in $\calV'(w,i)$ or $\calP'(i,j)$ are 
faces whose vertex on $S'_i$ comes {\em before} $c'_{i+1}$.  So new entries
(which must all be due to entry $c'_{i+1}$ since all other vertices on
$S_i'$ are also on $S_i$) can simply be appended and we maintain the order.

Now, consider the case where we scan an entire side $S_i$, except that $c'_i$
and/or $c'_{i+1}$ may already have been adjacent to the side and then do not 
get scanned.  If $c'_{i+1}$ was not yet on~$S_i$, then all new entries due
to it should be appended to $\calP(i,j)$.
If $c'_i$ was not yet on~$S_i$, then all new entries due to it should
be prepended (in reverse order).  So we must only handle the case where both
$c'_i$ and $c'_{i+1}$ were already on side $S_i$.  If $S_i$ only consists
of these two vertices, then no scanning happened and we are done; we
may therefore assume that $S_i$ has at least three vertices.  If
list $\calP(i,j)$ previously was empty, then we can simply scan.  So by
using the correct of the above approaches, we can create the correct order
in $\calP(i,j)$ as long as we ensure the following.

\begin{property}
\label{prop:cutting_pair}
Assume that, when recursing from $G$ to a subgraph $G'$, we scan an
entire side~$S'_i$.  Then either 
\begin{enumerate}[label=(\roman*)]
	\item $S'_i$ contains at most two vertices, or
  \item $c'_i$ is new to $S_i$, or 
  \item $c'_{i+1}$ is new to $S_i$, or
  \item all lists $\calP(i,j)$ are empty.
\end{enumerate}
\end{property}

\subsection{Details for the individual cases}

\begin{figure}[t]
\centering
\subcaptionbox{}{\includegraphics[scale=.75]{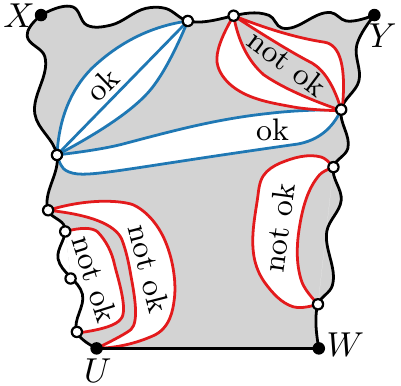}}
\hfil
\subcaptionbox{}{\includegraphics[scale=.9]{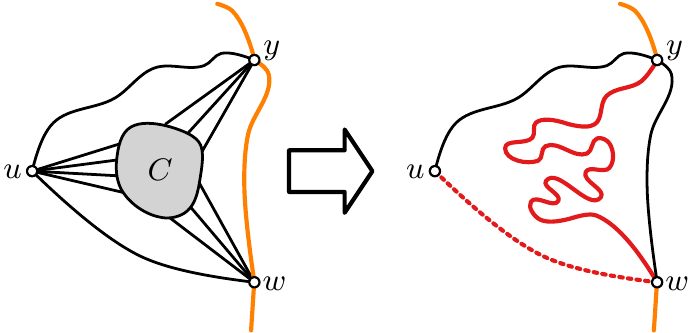}}
\hfil
\subcaptionbox{}{\includegraphics{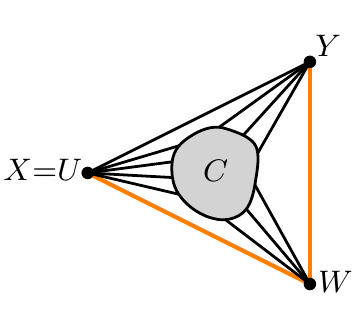}}
\repeatcaption{fig:c3csubs}{\ctcsubsCaption}
\label{fig:c3csubs-app}
\end{figure}

We now finally go through all cases and fill in some case-dependent details.
In particular, we must explain
\begin{enumerate}[label=(\Alph*)]
\item\label{step:test} how to test whether the case applies, and
\item\label{step:create} how to obtain the information to create the subgraphs,
	especially the necklace in Case~\ref{case:4}.
\end{enumerate}
We must also explain for each subgraph $G'$
\begin{enumerate}[label=(\Alph*)]
\setcounter{enumi}{2}
\item\label{step:corners} which corners of $G'$ correspond to which corners of $G$
	so that Property~\ref{prop:vertex_side} holds,
\item\label{step:face} why no divided face loses a side incidence (Property~\ref{prop:sides}),
\item\label{step:init} how to initialize the $\calP$-lists with the
	cutting pairs of $G$ that also exist in $G'$,
\item\label{step:sides} which vertices may be new to sides, and that we can scan them by
	scanning sides in such a way that Property~\ref{prop:scan_side} 
	and Property~\ref{prop:cutting_pair} holds.
\end{enumerate}
In all cases, we continue with the notations and assumptions that were
used in Section~\ref{subsec:proof}.

\subsubsubsection{Case~\ref{case:1} and the substitution trick (Fig.~\ref{fig:c3csubs-app}).}
\begin{enumerate}[label=(\Alph*)]
\item It is straightforward to test whether Case~\ref{case:1} applies, since we 
know the outer face.  

\item Observe that we only recurse into a subgraph if
we apply the substitution trick.  The subgraph is then 
$G':=G\setminus (\W,\X)$, i.e., we remove the edge that was the right side.  

\item The four corners of~$G$ remain the same, so clearly Property~\ref{prop:vertex_side} holds.

\item There are no divided faces.

\item Graph $G$ had no cutting pairs, so all $\calP$-lists as empty.
Hence Property~\ref{prop:cutting_pair} holds. 

\item The only vertices that could be new to a side are the ones 
that shared a face with $(\W,\X)$ in~$G$.
We should therefore scan along the entire side $S'_{\rs}$ of $G'$.%
\footnote{``$\rs$'' is a shortcut for ``the index $i$ such 
that $\W=c_i$, i.e., side $S'_i$ is the right side''.  Since this
can be looked up in $O(1)$ time from the location of $\U$ and
$\X$ in $\calC$, we use this convenient shorthand for this and
the other sides. }
All vertices in $S'_{\W\X}\setminus \{\W,\X\}$ are different
from $\W,\X$, hence new to $S'_{\rs}$.
\end{enumerate}

\subsubsubsection{Case~\ref{case:2} (Fig.~\ref{fig:case2-app}).} 

\begin{figure}[t]
\centering
\subcaptionbox{Case~\ref{case:2}}{\includegraphics[page=1,scale=.72]{case2}}
\hfil
\subcaptionbox{The path~$P_t$}{\includegraphics[page=2,scale=.72]{case2}}
\hfil
\subcaptionbox{The path~$P_b$}{\includegraphics[page=3,scale=.72]{case2}}
\hfil
\subcaptionbox{The path~$P$}{\includegraphics[page=4,scale=.72]{case2}}
\repeatcaption{fig:case2}{\casetwoCaption}
\label{fig:case2-app}
\end{figure}

\begin{enumerate}[label=(\Alph*)]
\item If Case~\ref{case:1} does not apply, then we check the list of cutting
pairs $\calP(\ls,\rs)$;  Case~\ref{case:2} applies if
and only if it is non-empty. 

\item Assuming $\calP(\ls,\rs)$ is non-empty, let $f^*$ be its
last entry (closest to $\V$).
Inspect $\calF(f^*,\ls)$ and $\calF(f^*,\rs)$ to find the
vertices $v,w$; taking the ones closer to the bottom  if there
are two possibilities.  
Split $G$ into $G_b$ and $G_t$ along the noose through~$v$,~$f^*$, and~$w$.
If $(v,w)$ was not an edge of $G$, then also add the virtual edge $(v,w)$ to 
split $f^*$ into two parts,~$f^*_b$ and~$f^*_t$, in the subgraphs.

\item $G_b$ retains corners $\V,\W$, and $G_t$ retains corners $\U,\X$;
both gain new corners~$v,w$.  We replace corners in $\calC$ (in this
and all other cases) so that the retained corners keep their spots and
the new corners  fill the emptied spots in ccw order; then, 
Property~\ref{prop:vertex_side} holds.

\item We must argue that $f^*$ loses no side incidences.
Face $f^*$ was incident
to both the left and the right side in $G$, and the same holds for both
parts in both subgraphs.  
Further,~$f^*$ is incident to the top side of~$G$
if and only if $f^*_t$ is incident to the top side of $G^+_t$, and~$f^*_b$ 
is always incident to the top side of $G^+_b$.  
Face~$f^*$ is incident to the bottom side of $G$
if and only if~$f^*_b$ is incident to the bottom side of~$G^+_b$, and
$f^*_t$ is always incident to the bottom side of~$G^+_t$.    So
all side incidences are retained for the parts of $f^*$ 
(and possibly there are new ones).

\item We chose the cutting pair such that $G_b^+$ has no other cutting pairs,
so the $\calP$-lists for $G_b^+$ are initialized empty.
Subgraph $G_t^+$ inherits the $\calP$-lists from $G$, 
except that we remove face $f^*$ from
$\calP'(\ls,\rs)$
if it had only one vertex each in $\calF(f^*,\ls)$ and $\calF(f^*,\rs)$.

\item Vertices $v,w$ may or may not be new to the bottom and the top side 
in $G_t$ and $G_b$, respectively.
So we scan along the entire sides~$S'_{\bs}$ in $G_t$ and $S'_{\ts}$ in $G_b$,
which both consist of exactly the edge $(v,w)$.
Since we scan along entire single-edge sides,
Properties~\ref{prop:scan_side} and \ref{prop:cutting_pair}(i) hold.
\end{enumerate}

\subsubsubsection{Case~\ref{case:3} (Fig.~\ref{fig:case3-app}).} 

\begin{figure}[t]
\centering
\subcaptionbox{}{\includegraphics[scale=.75,page=2]{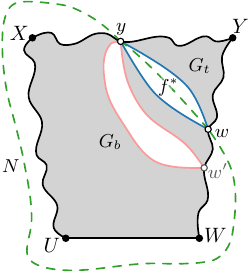}}
\hfil
\subcaptionbox{}{\includegraphics[scale=.75,page=8]{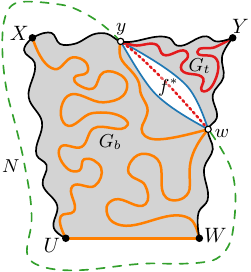}}
\hfil
\subcaptionbox{}{\includegraphics[scale=.75,page=5]{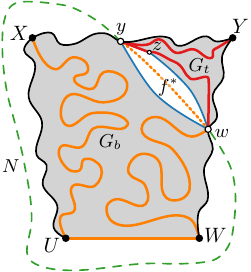}}
\hfil
\subcaptionbox{}{\includegraphics[scale=.75,page=6]{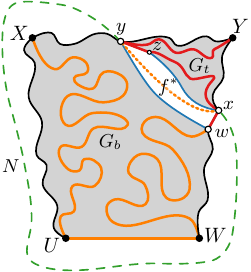}}
\repeatcaption{fig:case3}{\casethreeCaption}
\label{fig:case3-app}
\end{figure}

\begin{enumerate}[label=(\Alph*)]
\item To test whether Case~\ref{case:3} applies, we check whether $\calP(\rs,\ts)$
is non-empty.  Say it is, and let $f^*$ be its first face.
By the order of faces in $\calP(\rs,\ts)$, Case~\ref{case:3} can be applied
if and only if $f^*$ does not contain $\X$, which we can learn
in constant time from $\calF(f^*,\ts)$.

\item We find the cutting pair $\{y,w\}$ from $\calF(f^*,\ts)$
and $\calF(f^*,\rs)$.   
If either of these two lists contains two vertices, then we use the one that is clockwise
later on the outer face.  The resulting pair $\{w,y\}$ satisfies the
assumptions on the pair by the order of $\calP(\rs,\ts)$.
Split $G$ into~$G_b$ and $G_t$ along the noose through $y$, $f^*$, and $w$.
If $(y,w)$ was not an edge of~$G$, then also split~$f^*$ into 
two parts $f^*_b$ and $f^*_t$ in the subgraphs  $G_b^+$ and $G_t^+$, respectively.
We recurse in~$G_b^+$ first, which tells us whether Case~\ref{case:3a} or~\ref{case:3b} applies,
i.e., whether to insert $(y,w)$ into $G_t$ or not.  
In Case~\ref{case:3b}, we must further delete $w$ if it had degree 1 in $G_t$.

\item $G_b^+$ retains corners $\U,\V,\W$ and adds $y$ as its fourth corner.
In Case~\ref{case:3a}, $G_t^+$ retains corner~$\X$, and adds $y$, $w$, and again $w$.
In Case~\ref{case:3b-1}, $G_t$ retains corner $\X$ and adds $y$, its neighbor $z$ on
the outer face, and the vertex $w$.  In Case~\ref{case:3b-2}, $G_t'$ retains corner $\X$
and adds $y$, $z$, and the neighbor $x$ of $w$ as corners.

\item Let us first consider $f^*_b$.
Subgraph $G_b^+$ has corners $\U,\V,\W$, and $y$, with the same left and
bottom side as $G$.  Face $f^*$ is adjacent to the top and right side
of $G$, and the same holds for~$f^*_b$ in~$G_b^+$ due to vertex $y$.
Face $f^*$ cannot
have been adjacent to the left side of $G$ (else we would apply Case~\ref{case:2}).
If face $f^*$ was adjacent to the bottom side, then this must hence
have happened at $\W$, which makes $\{y,\W\}$ a cutting pair. By our
sorting of $\calP(\rs,\ts)$ hence $w=\W$, which means that $f^*_b$ then
is also adjacent to the bottom side at $w$.
So $f^*_b$ retains all side incidences that $f^*$ had.

As for $f^*_t$, this is involved in a recursion only if Case~\ref{case:3a} applies (else
it becomes part of the outer face).  Here, the corners of $G_t^*$ become
$\{\X,y,w,w\}$, which means that face $f_t^*$ touches all four sides since
it is incident to $y$ and $w$.

\item Observe first the status of the $\calP$-lists in $G$.  $\calP(\ls,\rs)$
was empty, else Case~\ref{case:2} would have applied.  
Any face in $\calP(\rs,\ts)$  is either~$f^*$ or belongs to~$G_t$, by our choice of~$f^*$.
Any face in $\calP(\ls,\ts)$  belongs to $G_b^+$ because $f^*$ has no
vertex on the left side (else Case~\ref{case:2} would apply).
Any face in $\calP(\bs,\ts)$  is incident to either $\V$ or $\W$.  Those
incident to $\V$ belong to $G_b$.  If any are incident to $\W$, then $w=\W$, 
these faces belong to $G_t$, and the first of them is $f^*$.  
So if $w=\W$, then we split $\calP(\bs,\ts)$
at the occurrence of $f^*$.  The $\calP$-lists for $G'$ can now be initialized
according to these insights, inheriting the list into the appropriate subgraph 
(after removing $f^*$ as needed) and initializing all other $\calP$-lists as empty.

\item In $G_b^+$, no new vertices become exterior, but $y$ is
new to the right side.  So we scan along edge $(w,y)$, which is part
of the right side of $G_b^+$.    Note that this fits
Property~\ref{prop:scan_side}(ii)  and Property~\ref{prop:cutting_pair}(iii).

For $G_t$, observe that $S'_{w\X}$ and $S'_{\X{y}}$
are part of their corresponding sides and need not be scanned, but
we scan the other two sides completely.    In Case~\ref{case:3a}, this
means that we scan $S'_{yw}$ and~$S'_{ww}$; since $(y,w)$ is an edge,
Properties~\ref{prop:scan_side}(i) and \ref{prop:cutting_pair}(i) hold.

In Case~\ref{case:3b-1}, we scan along side $S'_{yz}$ and $S'_{zw}$.  The former
is a single edge, 
so Properties~\ref{prop:scan_side}(i) and \ref{prop:cutting_pair}(i) hold.
The latter may have three or more vertices,
but all vertices other than $w$ were interior in $G$, hence are new
to the side.  So Property~\ref{prop:scan_side}(i) and
Property~\ref{prop:cutting_pair}(ii) hold, since $z$ is new to $S'_{zw}$.
The argument is the same for Case~\ref{case:3b-2}, replacing~`$w$' by~`$x$'.
\end{enumerate}

\subsubsubsection{Case~\ref{case:3'}.}
If Case~\ref{case:3} does not apply, then we test for Case~\ref{case:3'} 
by taking the last element $f^*$ of $\calP(\ls,\ts)$ and checking whether it
contains $\U$.  The treatment of this case is symmetric to Case~\ref{case:3}.

\subsubsubsection{Case~\ref{case:4} (Figs.~\ref{fig:case4a-app} and \ref{fig:case4b-app}).}

\begin{figure}[t]
\centering
\subcaptionbox{}{\includegraphics[page=1,scale=.75]{necklace}}
\hfil
\subcaptionbox{}{\includegraphics[page=2,scale=.75,trim={8mm 0 0 0},clip]{necklace}}
\hfil
\subcaptionbox{}{\includegraphics[page=3,scale=.75]{case4a}}
\hfil
\subcaptionbox{}{\includegraphics[page=4,scale=.75]{case4a}}
\repeatcaption{fig:case4a}{\casefouraCaption}
\label{fig:case4a-app}
\end{figure}

\begin{figure}[t]
  \centering
  \begin{subfigure}[b]{.3\textwidth}
    \centering
    \includegraphics[page=4,scale=.75]{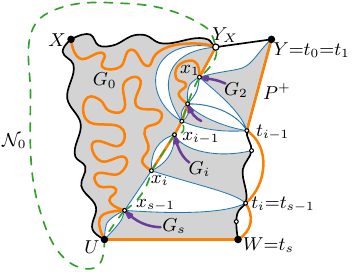}
    \caption{}
  \end{subfigure}
  \hfil
  \begin{subfigure}[b]{.3\textwidth}
    \centering
    \includegraphics[page=9,scale=1]{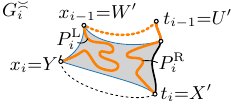}
    \bigskip
    \includegraphics[page=11,scale=1]{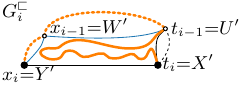}
    \caption{}
  \end{subfigure}
  \hfil
  \begin{subfigure}[b]{.3\textwidth}
    \centering
    \includegraphics[page=7,scale=.75]{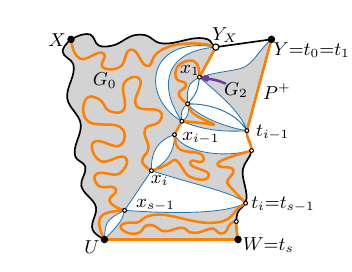}
    \caption{}
  \end{subfigure}
\repeatcaption{fig:case4b}{\casefourbCaption}
\label{fig:case4b-app}
\end{figure}

\begin{enumerate}[label=(\Alph*)]
\item If none of the previous cases applies, then we test condition
\starcondition by walking from $\W$ along the outer face to find out
whether to apply Case~\ref{case:4a} or Case~\ref{case:4b}.

\item Assume first that
we have to apply Case~\ref{case:4a}.  We then have to find a leftmost necklace
$\calN=\langle x_0{=}\U,\ldots,x_s{=}\W\rangle$, which is non-trivial.
We determine the vertices in order $x_s,x_{s-1},\dots,x_1$.  
Now do a left-first-search, starting from~$\W$ 
(see, e.g., de Fraysseix et al.~\cite{FMP95} for more on left-first searches).  However, we
search not only along edges, but also along all face-adjacent vertices,
and we advance the search only from vertices that are face-adjacent
to~$\X$.  See Algorithm~\ref{alg:necklacescana}.  We do some special
handling once we found a face containing $\U$, because this face (which
becomes face $f_1$ of the necklace) is already incident to the top side
and (as we will see) should not get scanned.

\begin{algorithm}[t]
  \caption{\protect\necklacescana}
  \label{alg:necklacescana}
  
  $(\fprev,x)\leftarrow (\mbox{outer face},\W)$ \;
  initialize the necklace $\calN$ with $\langle \fprev,x \rangle$ \;
  \While{$x \neq \U$}{
    \If{$x$ is not face-adjacent to $\U$}{
      \ForEach{face $\fnext$ incident to $x$, starting after $\fprev$ in clockwise order}{
        \ForEach{vertex $x'$ on $\fnext$, starting after $x$ in clockwise order}{
          \If{$x'$ is face-adjacent to $\X$}{
            append $\fnext,x'$ to $\calN$ \;
            $(\fprev,x) \leftarrow (\fnext,x')$ \;
            \Continue with next while iteration\;
          }
        }
      }
    }
    \Else(\tcp*[h]{Have reached the face of $\calN$ incident to $\U$.}){
      $\fnext\leftarrow$ last face in $V(x,\ts)$ \;
      \If{$\fnext$ does not contain $\X$}{
        append $\fnext,\U$ to $\calN$ \;
      }
      \Else (\tcp*[h]{to be leftmost, include all vertices along $\fnext$.}) {
        \Repeat{$x=X$} {
           $x'\leftarrow$ vertex after $x$ in clockwise order on $\fnext$ \;
           append $\fnext,x'$ to $\calN$ \;
           $x\leftarrow x'$\;
        }
      }
	{\bf return} the reverse of $\calN$ \;
    }
  }
\end{algorithm}

Since we perform a left-first-search, we find the leftmost path (where `path'
allows for face-adjacencies) that connects $\W$ to $\U$ along vertices that
are face-adjacent to $\X$.  Any $\W$-necklace also defines such a path, so
we find the leftmost $\W$-necklace.

We claim that the total time spent on finding $\calN$ is no more than
the time spent for updating the data structures for graph $G_0^+$ later.
First, we spend $O(1)$ time per vertex that we found,
hence $O(s)$ time overall.   Since $G_0^+$ has $s$ new edges on the
outer face, this is accounted for by the $O(|E_X(G_0^+)|)$ time that we
already counted for $G_0^+$.  Second, we 
spend $O(\deg(f'))$ time on scanning an interior face $f'$ incident to vertex $x$.
This happens only if $f'$ contains neither $\U$ nor $\X$, because otherwise
we immediately find the next vertex of the necklace, and this was counted
above.  Since $x$ becomes a vertex on the top side
of $G_0^+$, and $f'$ contains neither $\U$ nor $\X$ (and hence no vertex
on the top side of $G$), it is newly adjacent to the top side and hence
will get scanned during the update.

\medskip
For Case~\ref{case:4b}, finding the $\V$-necklace is done similarly, 
except that we search from the top downward, rather than from the bottom 
upward, 
which requires exchanging `clockwise' by `ccw' in the order of scan.
See Algorithm~\ref{alg:necklacescanb}.
Recall that $x_1$ is not arbitrary in Case~\ref{case:4b};  it is the clockwise
neighbor after $\X$ in~$L(\X_\U)$, and we want the leftmost necklace
containing it.

\begin{algorithm}[t]
  \caption{\protect\necklacescanb}
  \label{alg:necklacescanb}
  
  initialize the necklace $\calN$ with $\langle \mbox{outer face},\X_\U \rangle$ \;
  $(\fprev,x)\leftarrow$ first two entries in $L(\X_\U)$ \tcp*{This identifies $x_1$ correctly.} 
  append $\fprev,x$ to $\calN$ \;
  \While{$x \neq \V$}{
    \If{$x$ is not face-adjacent to $\V$}{
      \ForEach{face $\fnext$ incident to $x$, starting after $\fprev$ in ccw order}{
        \ForEach{vertex $x'$ on $\fnext$, starting after $x$ in ccw order}{
          \If{$x'$ is face-adjacent to $S_{\rs}$}{
            append $\fnext,x'$ to $\calN$ \;
            $(\fprev,x) \leftarrow (\fnext,x')$\;
            \Continue with next while iteration\;
          }
        }
      }
    }
    \Else(\tcp*[h]{Have reached the face of $\calN$ incident to $\V$.}){
      $\fnext\leftarrow$ first face in $V(x,\bs)$ \;
      \If{$\fnext$ does not contain vertex of $S_{\rs}$}{
        append $\fnext,\U$ to $\calN$ \;
      }
      \Else (\tcp*[h]{to be leftmost, include all vertices along $\fnext$.}) {
        \Repeat{$x=\V$} {
           $x'\leftarrow$ vertex after $x$ in ccw order on $\fnext$ \;
           append $\fnext,x'$ to $\calN$ \;
           $x\leftarrow x'$\;
        }
      }
	{\bf return} $\calN$ \;
    }
  }
\end{algorithm}

Since we are doing a (modified) right-first-search from $x_1$, we will find
the leftmost necklace that includes $x_1$ as desired.  
The running time is $O(s)+O(\sum_{f'} \deg(f'))$,
where we sum over all those interior faces $f'$ that are incident to
$x_i$ (for some $1\leq i\leq s-1$), and did not contain~$\V$ or $\W$
(else we immediately have found the next vertex of the necklace).
Any such face~$f'$ was hence not incident to the bottom side of $G$, but
is incident to the side $\langle x_s,\dots,x_1\rangle$
that becomes the corresponding  side in $G_0^+$.  
So, as above, the time for the necklace scan is a constant overhead for
the time for the scanning done to update the data structures.

\medskip

Note that, in both cases,
once we found the necklace $\langle x_0,f_1,x_1,f_2,\dots,f_{s},x_s\rangle$,
we can also easily find the vertices $t_i$ for $i=1,\dots,s-1$ ($t_0$ 
and $t_s$ are determined from the case).  Namely, to find~$t_i$,
let $f'$ be the last face in $\calV(x_i,\rs)$ (this exists by choice
of~$x_i$); then,~$t_i$ is one of the vertices in $\calF(f',\rs)$, using the
one that comes ccw later on the outer face if there are two.
This gives all the required information to determine the subgraphs to
recurse in.

\item In Case~\ref{case:4a},  graph $G_0^+$ retains corners $\U,\V,\W$, and
adds $\W$ again as a corner.    In Case~\ref{case:4b}, graph $G_0^+$
retains corners $\U,\V$ (though their order ``rotates''; 
$\U$ was the top-left corner in $G$ but becomes the top-right
corner in $G_0^+$) and adds $x_1$ and $x_0=\X_\U$ as corners.

Any subgraph $G_i$ (for $i=1,\dots,s$ and in both Case~\ref{case:4a} and~\ref{case:4b})
has four corners $x_{i-1},x_i,t_i,t_{i-1}$.  We let $t_{i-1}$ take
the place of $\X$ in $\calC$ (with this, side $S'_{t_i,t_{i-1}}$ of
$G_i$ corresponds to side $S_{\W\X}$ of~$G$ as required for
Property~\ref{prop:vertex_side}).  

As always, corners that are neither retained nor explicitly assigned
a place in $\calC$ are filled in as to maintain the ccw order in $\calC$.

\item We need to argue that any face $f$ that was divided in Case~\ref{case:4} 
does not lose any side incidences.  

Assume first that some part of $f$ belongs to graph $G_0^+$ in Case~\ref{case:4a}.  
	Any incidence of $f$ with the left or bottom side of $G$
	is carried over since $G_0^+$ inherits these sides.
	So we only have to worry about incidences of $f$ with the right or 
	top side of $G$.
	But $f$ cannot be incident to~$\X$, else by Claim~\ref{cl:virtual}
	it would have no part in $G_0^+$.  So if $f$ was incident to the 
	top side,
	then by~\starcondition it must contain $\U$, and it retains this vertex on the top side
	of~$G_0^+$.  Similarly, if $f$ was adjacent to the right side, then by \starcondition
	it must contain $\W$, and it retains this vertex on the right side of $G_0^+$.

Assume next that $f$ has parts in graph $G_0^+$ in Case~\ref{case:4b}.  
	Any incidence of $f$ with the left side of $G$
	is carried over since $G_0^+$ inherits this side.
	Face $f$ cannot contain any vertex on the right side of $G$ (and in 
	particular not $\X$ or $\W$), else by 
	Claim~\ref{cl:virtual} it would have no part in $G_0^+$.   So~$f$ was
	not incident to the right side of~$G$.  Furthermore, if $f$ was incident
	to the top side of $G$, then this incidence happened at a vertex $z\neq \X$,
	so $z$ is also in~$G_0^+$ and hence $f$ retains this side incidence.  If 
	$f$ was incident to the bottom side $(\V,\W)$ of~$G$, 
	then it is incident to $\V$ which belongs to $G_0^+$,
	and so $f$ retains this side incidence. 

Assume next that $f$ has parts in graph $G_i$ for some $1\leq i\leq s$,
	and $t_{i-1}=t_i$.  
	We recurse in $G_i$ only if we apply the substitution trick,
	in case of which edge $(x_{i-1},x_i)$ is removed.  In particular, the 
	face $f_i$ that was split along $(x_{i-1},x_i)$ has no part in the 
	graph $G_i^+:=G_i\cup (x_{i-1},t_i)\cup (x_i,t_i)$ 
	in which we recurse.  
	So $f\neq f_i$, hence $f$ must have been split by $(x_{i-1},t_i)$ or 
	$(x_i,t_i)$ and $f$ contains $t_i$.  Because $t_i$ becomes both the 
	top-right and the bottom-right corner in $G_i^+$, face~$f$ (which is 
	also adjacent to one of the left corners~$x_{i-1}$ and~$x_i$)
	is incident to all four sides of $G_i^+$.  In particular, no 
	side incidence has been lost.

Assume finally that $f$ has parts in graph $G_i$ for some $1\leq i\leq s$,
	and $t_{i-1}\neq t_i$.    (This in particular implies that we are in Case~\ref{case:4b}.) 
	We argue that the side incidences are preserved for
	the resulting graph $G_i^\sqsubset$; then, they clearly also hold 
	for $G_i^\asymp = G_i^\sqsubset \setminus (x_{i-1},x_i)$ which
	has even larger sides.

	Graph $G_i^\sqsubset$ has as its four sides the edges $(x_{i-1},t_{i-1})$,
	$(x_{i-1},x_i)$, $(x_i,t_i)$, and path~$S_{t_it_{i-1}}$.  Any face that was 
	divided and has parts in $G_i^\sqsubset$ contains one of these three edges, hence immediately is adjacent to
	two corners and with them three sides, and we only need to worry about the one remaining side.
	
	Assume that $f$ was divided by $(x_{i-1},x_i)$, in case of which we only need
	to worry if $f$ had vertices on the right side of $G$.   But then, by Claim~\ref{cl:virtual},
	edge $(x_{i-1},x_i)$ was not virtual and~$f$ was not divided by this edge, a contradiction. 

	Assume that $f$ was divided by $(x_{i-1},t_{i-1})$, in case of which we only need
	to worry if $f$ had a vertex $z$ on the bottom side of $G$.   
  If $z=\V$, then $\{z,t_{i-1}\}$ would be a cutting pair by $t_{i-1}\neq t_i$, 
  and we would have applied Case~\ref{case:2}.  So $z=\W$, which implies
	$t_i=\W$ (else we would have a cutting
	pair within the right side of $G$). So $f$ contains $t_i$, and hence is incident 
  to all sides of $G_i^\sqsubset$. 

Assume that $f$ was divided by $(x_{i},t_{i})$, in case of which we only need
to worry if $f$ had a vertex $z$ on the top side of $G$.   We claim that $f$ 
contains $\X$.  Namely, if $z\neq \X$, then $\{z,t_{i}\}$ would be a cutting pair 
by $t_{i-1}\neq t_i$.  Since we did not apply Case~\ref{case:3}, therefore $f$ contains $\X$.
This implies $t_{i-1}=\X$ (else we would have a cutting
pair within the right side of $G$). So $f$ contains $t_{i-1}$, and hence is incident 
to all sides of $G_i^\sqsubset$.

\item Since we applied neither Case~\ref{case:2} nor Case~\ref{case:3}, there 
are only two possible cutting pairs in~$G$: The neighbors of $\U$, 
if $\deg(\U)=2$, and the neighbors of $\X$, if $\deg(\X)=2$.

If $\deg(\U)=2$ or $\deg(\X)=2$, then the top side is not a single edge
(else we would be in Case~\ref{case:1} or~\ref{case:2}), so Case~\ref{case:4b} 
applies.  If $\deg(\U)=2$, then 
its neighbors lie on the left and top side of $G$
(recall that we assumed $\U\neq \V$).
So they belong to $G_0^+$ and form a cutting pair there.
Initialize
the appropriate $\calP$-lists of $G_0^+$ with the interior face at $\U$.

If $\deg(\X)=2$, then let $f_\X$ be its unique interior face. 
Let~$\X_\U$ be the neighbor of~$\X$ on the top side; this exists because $\U\neq\X$.
If $\X=\W$, then the neighbors of~$\X$ would be $\V$ and $\X_\U$ and form a cutting
pair on the left and top side.  Since Case~\ref{case:3'} does 
not apply, therefore $f_\X$ also contains $\U$, and therefore $\X_\U=\U$ (else
there is a cutting pair within the top side).  But then $(X,U)$ must
be an edge (else there would be a cutting pair within the left side),
which means that the outer face is triangle and we are in the base case.
So we know $\X\neq \W$.  Let 
$\X_\W$ be the neighbor of $\X$ on the right side (possibly $\X_W=\W$).
Cutting pair $\{\X_\U,\X_\W\}$ is inherited by the subgraph~$G_i$ that contains
the edge $(\X_\W,\X)$.  We initialize
the appropriate $\calP$-lists of $G_i$ with the face $f_\X$.

All $\calP$-lists not initialized by the above are initialized empty.

\item We scan in all subgraphs along all entire sides that contain newly 
exterior edges.  Such sides are single-edge sides, with the exception of side
$\langle x_0,\dots,x_s \rangle$ (in Case~\ref{case:4a}) and
$\langle x_1,\dots,x_s \rangle$ (in Case~\ref{case:4b}).  But $x_1,\dots,x_{s-1}$
were interior and hence new to the side, so 
Property~\ref{prop:scan_side}(i) holds.
Also, the $\calP$-lists are empty in Case~\ref{case:4a}, and in Case~\ref{case:4b} all scanned
sides are single edges or the side
$\langle x_1,\dots,x_s \rangle$, for which corner $x_1$ is new to the side.
So Property~\ref{prop:cutting_pair} holds.
\end{enumerate}

\bigskip
With this, we have explained how to do all steps in all cases.  Note in
particular that computing the necklace takes no more time than what was
spent on scanning vertices and faces.  Note further that
Property~\ref{prop:sides} holds, hence the running time is linear
(in the degrees of vertices and faces touched by the output path $P$)
as desired.
		
%%%%%%%%%%%%%%%%%%%%%%%%%%%%%%%%%%%%%%%%%%%%%%%%%%%%%%%%%%%%%%%%%%%%%%%%
\section{Finding binary spanning trees}
\label{sec:3trees}

In this section, we prove Theorem~\ref{thm:3tree}, i.e., we show that
any 3-connected planar graph has a binary spanning tree and it can
be found in linear time.  
This was known before \cite{Bie-Barnette,Strothmann-PhD}, but we give
a different proof here as a warm-up for Section~\ref{sec:2walks} and
because our binary spanning tree has some other interesting properties.
We need a slightly
stronger statement for the induction step to go through; in particular,
we must be allowed to fix one leaf and one edge of the tree, as long as
they are exterior to $G$.

\begin{lemma}\label{lem:binary-tree}
Let $G$ be a plane graph with distinct vertices $\U,\X$ on
the outer face.  Let $(\V,\W)\neq(\U,\X)$ be an edge on the outer face.
If $G$ is corner-3-connected with respect to $\U,\V,W,\X$,
then it has a binary spanning tree $T$ such that when rooting $T$ at $\U$
\begin{itemize}
\item $\X$ is a leaf and $(\V,\W)$ is an edge of $T$,
\item a vertex $v$ has two children in $T$ only if it is an interior
	vertex and part of a cutting triplet $\{v,w,x\}$ of $G$;
	moreover, one of the subtrees of $v$ contains exactly the
	vertices interior to $\{v,w,x\}$.\
\end{itemize}
The tree $T$ can be found in linear time.
\end{lemma}
\begin{proof}
We proceed by induction on the number of vertices, with an inner
induction on the number of exterior vertices.  
In the base case, $G$ is a triangle, say $\langle X{=}U,W,Y\rangle$,
and the path $\langle X,W,Y \rangle$ satisfies all properties.  Now, 
assume that $n\geq 4$, and let $P$ be a $\Tint$-path from $X$ to $Y$ that
contains $(\V,\W)$.  By Theorem~\ref{thm:linearTime}, we know that this path
can be found in $O(\sum_{f\in F(P)} \deg(P))$ time.

We expand $P$ into a binary spanning tree $T$ by expanding each $P$-bridge $C$ into
a subtree to be attached at its representative $\sigma(C)$; see Fig.~\ref{fig:binary-tree}.  Formally,
let $\{x,w,y\}$ be the attachment points of $C$, with $x=\sigma(C)$.
Define $C^+$ as in the substitution trick, i.e., $C^+=G[C] \cup  \{(x,w),(w,y)\}\setminus \{(x,y)\}$.
This satisfies $c3c(x,w,y)$ and either has fewer vertices or has a bigger
outer face. So, by induction, we can find a binary spanning tree~$T_C$ 
in $C^+$ that, when rooted at $x$, has $y$ as a leaf and contains
edge $(w,y)$; in particular,~$w$ is the parent of $y$.  Therefore ,
$T_C\setminus \{w,y\}$ is a spanning tree of $C\cup \{x\}$, and
$x$ has only one child in it since it is exterior in $C^+$.  We merge
this subtree into $P$ at $x$.  Repeating this for all $P$-bridges gives
the final tree $T$.  Clearly, the first property holds for $T$ since it
held for~$P$ already.  If a vertex $v$ has two children in $T$, then
necessarily $v=\sigma(C)$ for some $P$-bridge~$C$.   Therefore, $v$ is
interior since $P$ is a $\Tint$-path.  The rest of the second property
follows immediately since one subtree of $v$ was obtained by merging
$T_C\setminus \{w,y\}$.

\begin{figure}[t]
\centering
{\includegraphics[scale=.75,page=1]{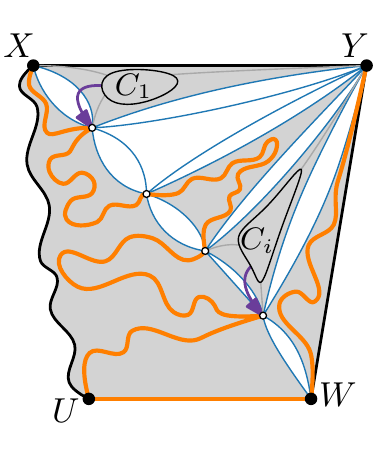}}
\hfil
{\includegraphics[scale=.75,page=3]{binary-tree}}
\hfil
{\includegraphics[scale=.75,page=2]{binary-tree}}

\caption{Illustration for Lemma~\ref{lem:binary-tree}.}
\label{fig:binary-tree}
\end{figure}

As for the running time, observe that to find $T$ we first find a $\Tint$-path $P$
in $G$, then a $\Tint$-path $P_C$ for each $P$-bridge $C$, then a $\Tint$-path
for each $P_C$-bridge, and so on recursively.  Crucial for an efficient
running time is that the time to find $P$ is proportional to
$\sum_{f\in F(P)} \deg(f)$, i.e., faces that are strictly interior to $C$
do not contribute to this running time.  Secondly, in the time for finding $P$
we can also initiate
all the data structures as they are needed for $C^+$ (because this is
the same graph as was used for the substitution trick).    This saves having
to re-initialize data structures by scanning when recursing into~$C^+$.
Therefore, exactly as for Theorem~\ref{thm:linearTime}, one argues that all
faces are scanned a constant time, giving an overall linear running time.
\end{proof}

\section{Finding 2-walks}
\label{sec:2walks}

In this section, we prove Theorem~\ref{thm:2walk}, i.e.,
we show that every 3-connected planar graph has a 2-walk.  
This was originally shown by Gao  and Richter \cite{GR94}
(with later improvements together with Yao~\cite{GRY95,GRY06}).
As in their proofs, we show instead the existence of a
{\em 2-circuit}, i.e., a circuit in the graph
that visits every vertex at least once and at most twice; a 2-walk
can then be obtained by deleting one edge.
Gao and Richter argue that a 2-circuit can be found by starting with
a $\TSDR$-path $P$ (with an extra edge added to close it to a circuit) and
then recursively finding a 2-circuit in each $P$-bridge $C$ and attaching it
at the representative of~$C$.

While this is in spirit very similar to how we found binary spanning trees
(Section~\ref{sec:3trees}),
and so one would expect that this can be done in linear time using our algorithm
for Tutte paths,
there are some details that make this more complicated.  In particular, Gao
and Richter needed the ability to find a Tutte path in an internally
3-connected graph (which they called a \emph{circuit graph}).    Not every
such graph satisfies corner-3-connectivity with at most~4 corners,
so that we cannot apply Lemma~\ref{lem:T_int} directly to find those
Tutte paths.  Therefore, we argue how to find the 2-circuits in linear time in
two steps: first we argue how to find Tutte paths in internally 3-connected graphs in linear
time, and then we repeat the proof of Gao, Richter, and Yao~\cite{GRY95} to find 2-circuits, and argue
that this also takes linear time.

\subsection{Internally 3-connected graphs}

We first expand Lemma~\ref{lem:T_int} to 
graphs that are internally 3-connected.
We consider here only the special case where $(\U,\X)$ is an edge;
the existence of a $\TSDR$-path could be shown even in the absence of
this edge, but if the edge exists, then we can restrict the representatives.

\begin{lemma}
\label{lem:circuit}
Let $G$ be a plane graph with distinct edges $(\U,\X)$ and $(\V,\W)$ on the outer face.
If $G$ is internally 3-connected, 
then it has a $\TSDR$-path $P$ that begins at $\U$, ends at $\X$, and contains $(\V,W)$.
Furthermore, $\U$ and $\W$ are not used as a representative.  This path can be found in time
$\sum_{f\in F(P)} \deg(f)$.
\end{lemma}
\begin{proof}
We may assume that $\U,\V,\W,\X$ occur in ccw order, the other case is symmetric.
Consider first the left side (with respect to the four corners $\U,\V,\W,\X$), 
and enumerate it as $\langle X{=}u_0,u_1,\dots,u_\ell{=}\V\rangle$.
There may now be cutting pairs within the left side.  Let $i_1$ be minimal such that
$\{u_{i_1},u_{j_1}\}$ is a cutting pair for some $j_1>i_1$,  and assume that $j_1$ has
been chosen maximal among all such cutting pairs.  Repeating, for $k=2,3,\dots$, let
$i_k\geq j_{k-1}$ be minimal such that $\{u_{i_k},u_{j_k}\}$ is a cutting pair for some $j_k>i_k$,
and choose $j_k$ to be maximal among all such cutting pairs. 

Now remove all these cutting pairs by removing, for each $k=1,2,\dots$ the cut component~$C_{k}$
of $\{u_{i_k},v_{i_k}\}$ that does not contain the
right side of $G$, and adding a virtual edge $(u_{i_k},u_{j_k})$ if it did not
yet exist in $G$.  Also set $\sigma(C_{k})=u_{j_k}$, i.e., assign as representative of~$C_j$ 
the vertex of the cutting pair that is closer to $\V$ along the left side.  Note in particular that 
any cutting pair receives a distinct representative on the outer face, and 
$\U$ will not be assigned as a representative.  (Neither will $\W$, since it is not on the left side.)
See Fig.~\ref{fig:circuit-graph}.

\begin{figure}[t]
\centering
{\includegraphics[scale=.75,page=1]{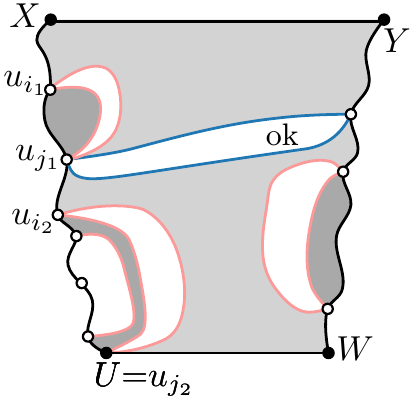}}
\hfil
{\includegraphics[scale=.75,page=3]{circuit-graph}}
\hfil
{\includegraphics[scale=.75,page=2]{circuit-graph}}

\caption{Illustration for Lemma~\ref{lem:circuit}.}
\label{fig:circuit-graph}
\end{figure}

Similarly, remove all cutting pairs on the right side, replacing the
cut component by a virtual edge and assigning as representative the vertex of
the cutting pair that is closer to~$\X$.
Again, all representatives are distinct and
$\U,\W$ will not be assigned as a representative.

Note that there can be no cutting pair within the top or bottom side, since these sides
are single edges $(\U,\X)$ and $(\V,\W)$.
Thus, after having removed cutting pairs on the left and right side, the resulting graph $G'$ is 
corner-3-connected with respect to $\U,\V,\W,\X$.  Apply Lemma~\ref{lem:T_int} to $G'$
to find a $\Tint$-path $P'$ of $G'$.  None of the representatives used by $P'$ conflicts
with representatives assigned to cutting pairs, because the latter are all exterior.

If $P'$ uses no virtual edge, then it can serve as the desired $\TSDR$-path.  Otherwise,
for every virtual edge $(u_{i_k},u_{j_k})$ used by $P'$, we replace this edge with
a path through $C_k$.  More precisely, recursively find a $\TSDR$-path $P_k$ in 
$C_k^+:=G[C_k]\cup\{(u_{i_k},u_{j_k})\}$
that begins at~$u_{i_k}$, ends at~$u_{j_k}$, uses some other exterior edge, 
and does not use $u_{i_k}$ as a representative.    This exists since 
$C_k^+$ is again internally 3-connected, and it does not use the (possibly virtual) edge
$(u_{i_k},u_{j_k})$ since this edge connects the ends of $P_k$.  Then substitute~$P_k$ in
place of edge $(u_{i_k},u_{j_k})$ in $P$.  

It remains to analyze the time complexity.  The bottleneck is finding the cutting
pairs within each side.  Let us assume that we have scanned the outer face of $G$
to mark each vertex with the side that it belongs to.
Recall that we initiate the data structures for $G$ by scanning at each exterior vertex
along all its incident faces.  Say we do this scan in ccw order, and are currently
scanning vertex $u_{i_k}$, starting clockwise after the outer face.  If, during this
scan, we encounter some vertex $u_{j_k}$ that is on the same side, then we have
found a cutting pair that must be removed; we immediately split the graph at the
currently scanned face (this removes the cut component $C_k$) and add the edge 
$(u_{i_k},u_{j_k})$ if it did not exist yet.    Since we scanned faces in clockwise order,
index $j_k$ automatically is maximized.  Note that we had not yet scanned
any faces that are interior to $C_k$, and we will not do so
(at least not for the purpose of finding $P'$) since $C_k$ is now removed from $G$.

Let~$G'$ be the resulting graph after removing all these cutting pairs.
We can find the $\Tint$-path in time
$O(\sum_{f\in F'(P')} \deg(f))$, where $F'(P')$ denotes the faces incident to $P'$
in~$G'$.    This proves the claim if we need not recurse into
any cut component $C_k^+$.  If we do recurse in $C_k^+$, then we must update the
data structure for it.  However, we are free to choose which exterior edge of
$C_k^+$ must be visited by $P_k$; we can choose this to be the one incident to $u_{j_k}$
so that the entire outer face of $C_k^+$ (except for $u_{j_k}$) becomes the right side
and can therefore inherit the side-marking that it had obtained in $G$.  Then we
recurse in $C_k^+$, which takes time at most $O(\sum_{f\in F_k(P_k)} \deg(f))$, where
$F_k(P_k)$ are the faces in $C_k^+$ that are incident to $P_k$.  This repeats for 
possibly many $P'$-bridges, but since the interior faces of these $P'$-bridges are
distinct and $F'(P')\cup \bigcup_k F_k(P_k)=F(P)$,
the claim on the running time follows.
\end{proof}

\subsection{2-circuits}

We now use Lemma~\ref{lem:circuit} to prove Theorem~\ref{thm:2walk}, i.e., to find 
2-circuits in planar 3-connected graphs.
The proof idea is exactly as in \cite{GRY95}, we repeat it here to argue that it can
be implemented in linear time.      Again we need a slightly stronger statement
to make the induction hypothesis work.

\begin{lemma}
Let $G$ be an internally 3-connected plane graph with exterior vertices $\U\neq \W$.
Then $G$ has a 2-circuit $P$ that visits $\U$ and $\W$ exactly once, and for which
a vertex $v$ is visited twice only if $v$ is part of a separating triplet
or a cutting pair.  $P$ can be found in linear time.
\end{lemma}
\begin{proof}
Let $\X$ and $\V$ be the clockwise neighbor of $\U$ and $\W$ on the outer face.
Find a $\TSDR$-path from $\U$ to $\X$ through $(\V,\W)$ using
Lemma~\ref{lem:circuit} and complete it to a circuit $P'$ by adding edge $(\U,\X)$.
If $P'$ visits all vertices of $G$, then we are done.  Otherwise, there are
some $P'$-bridges, and we expand $P'$ into a 2-circuit by recursively finding
2-circuits in the $P'$-bridges and merging them into $P'$.  
Formally, let $C$ be a $P'$-bridge.  We distinguish cases by whether $C$ has two
or three attachment points.

\smallskip\noindent{\bf $C$ has two attachment points $\{u_i,u_j\}$; see Fig.~\ref{fig:2walk-two}:}  Let us
assume that $\sigma(C)=u_j$ is the representative of $C$.  Let $C^+$
be the graph induced by $C\cup \{u_i,u_j\}$ with edge $(u_i,u_j)$ added if
it did not exist in $G$.    We know that $C^+$ is internally 3-connected,
therefore (as argued in \cite{GR94}) $C^+\setminus u_i$ is a so-called plane chain of blocks
(defined below).

\smallskip\noindent{\bf $C$ has three attachment points $\{x,w,y\}$; see Fig.~\ref{fig:2walk-three}:}  Let use
assume that $\sigma(C)=x$ is the representative of $C$.  Let $C^+$
be graph $G[C]\cup \{(x,w),(w,y),(y,x)\}$.  We know that $C^+$ is 3-connected,
therefore $C^+\setminus \{w,y\}$ is again a plane chain of blocks~\cite{GR94}.

\smallskip
In both cases, hence the graph induced by $C\cup \sigma(C)$ is a {\em plane chain of blocks}. 
This means that it is the union of internally 3-connected graphs $B_1,\dots,B_k$
where
\begin{enumerate}[label=(\arabic*)]
\item $B_1$ contains~$\sigma(C)$ (set $x_1:=\sigma(C)$), 
\item for $\ell=2,\dots,k$ graphs $B_{\ell-1}$ and $B_\ell$ 
  have one exterior vertex $x_\ell$ in common, 
\item $\{x_1,\dots,x_k\}$ are mutually distinct, and 
\item the graphs are disjoint otherwise.  
\end{enumerate}

\begin{figure}[t]
\centering
\begin{subfigure}[b]{.3\textwidth}
  \centering
  \includegraphics[scale=.75,page=2]{2walk}
  \caption{Two attachment points}
  \label{fig:2walk-two}
\end{subfigure}
\hfil
\begin{subfigure}[b]{.3\textwidth}
  \centering
  \includegraphics[scale=.75,page=1]{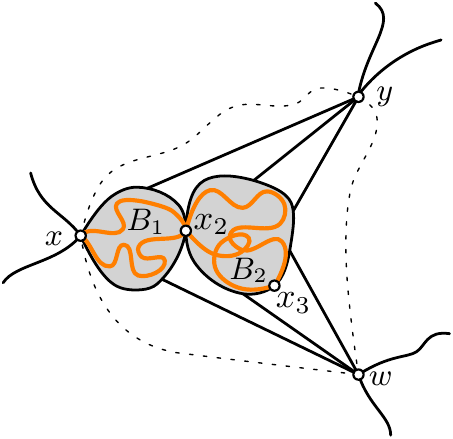}
  \caption{Three attachment points}
  \label{fig:2walk-three}
\end{subfigure}
\caption{Finding a 2-circuit in a $P'$-bridge.}
\label{fig:2walk}
\end{figure}

Recursively find 2-circuits $P_1,\dots,P_k$ for $B_1,\dots,B_k$ such that $P_\ell$
visits $x_\ell$ and $x_{\ell+1}$ exactly once (we set $x_{k+1}\neq x_k$ to be an arbitrary
exterior vertex of $B_k$).  Set $P_C$ to be $P_1\cup \ldots \cup P_k$,  i.e., the union
of the edges of the circuits. This has even degrees and hence is again a circuit.
Vertex $x_i$, for $i>1$, is visited once each by $P_{i-1}$ and $P_i$, so visited twice
by $P_C$.  Observe that $x_i$ is a cut vertex of $G[C\cup \sigma(C)]$, hence either part
of a cutting pair $\{x_i,u_{i}\}$ or a cutting triplet $\{x_i,w,y\}$.  
So $P_C$ is a 2-circuit of $G[C\cup \sigma(C)]$ that visits $\sigma(C)$ once and
satisfies the requirements on vertices that are visited twice.

Adding $P_C$ to $P'$ hence
gives a circuit of $G$ that now visits all vertices in $C$ at least once.  This
new circuit visits vertex $\sigma(C)$ repeatedly, from $P'$ and once from $P_C$.  However,
vertex~$\sigma(C)$ was visited exactly once by $P'$ initially, and it is used as a
representative for only one $P'$-bridge, so in the final circuit $P$ (obtained after applying
the substitution for all $P'$-bridges) $\sigma(C)$ is visited twice. Hence $P$ is a 2-circuit.
Furthermore, since $\sigma(C)\neq \U,\W$ is an attachment vertex of $C$
and  $P$-bridges have at most three attachment points, all properties hold.

It remains to analyze the running time, and in particular, how to find the chain of blocks
efficiently.    (We cannot afford to run the standard algorithms to find 2-connected
components, because this takes linear time, which may result in overall quadratic time if
we do it repeatedly in the recursions.)  We explain this
only for the case of three attachment points, the other case is similar.  Assume as
before that $C$ has attachment points $\{x,w,y\}$ with $\sigma(C)=x$.  Recall that $x,w,y\in P$.
We can therefore afford (as overhead to the time to find $P$) to scan those faces incident
to $w,y$ that are inside $C^+$.  So we can mark all interior vertices of $C^+$ that are
face-incident to $w$, and (with a different mark) all interior vertices of $C^+$ that
are face-incident to $y$.  Any vertex that is marked from
both $w$ and $y$ is a cut vertex of $C$ (or vertex $x$), and should hence become one of the vertices
$x_1,\dots,x_k$. These vertices are discovered in order if we scan the incident faces in order.
All vertices that are face-incident to one of $w$ and $y$ (but not the other) become
the outer face of $B_1,\dots,B_k$.  
We have hence found
the split into subgraphs to recurse in without scanning any faces inside those subgraphs.
As for Lemma~\ref{lem:circuit}, one argues that all recursions overall take time
proportional to the degree of all faces, which is linear.
\end{proof}

%%%%%%%%%%%%%%%%%%%%%%%%%%%%%%%%%%%%%%%%%%%%%%%%%%%%%%%%%%%%%%%%%%%%%%%%
\section{Outlook}

In this paper, we improved on a very recent result that shows that
Tutte paths in planar graphs can be found in quadratic time. We gave
a different existence proof which leads to a linear-time
algorithm.  For 3-connected planar graphs, we obtain not only a
Tutte path, but furthermore endow it with a system of distinct
representatives, none of which is on the outer face.  With this,
we can also find 2-walks and binary spanning trees in 3-connected 
planar graphs in linear time.

The main remaining questions concern how to find Tutte path in other
situations or with further restriction.  For example, Thomassen
\cite{Thomassen83} and later Sanders \cite{Sanders97} improved Tutte's
result and showed that we need not restrict the ends of the Tutte path
to lie on the outer face.   These paths can be found in quadratic time
\cite{SS-ICALP18}.  But our proof does not seem to carry over to the
result by Sanders, because the ends of the path crucially must coincide
with corners of the graph.  Can we find such a path in linear time?

Furthermore, the existence
of Tutte paths has been studied for other types of surfaces
(see, e.g., Kawarabayashi and Ozeki~\cite{KO13} and the references therein).  
Can these Tutte paths be found in polynomial time, and preferably, linear time?

%%%%%%%%%%%%%%%%%%%%%%%%%%%%%%%%%%%%%%%%%%%%%%%%%%%%%%%%%%%%%%%%%%%%%%%%
\bibliographystyle{plainurl}
\bibliography{tutte-paths}

%%%%%%%%%%%%%%%%%%%%%%%%%%%%%%%%%%%%%%%%%%%%%%%%%%%%%%%%%%%%%%%%%%%%%%%%

\end{document}